\documentclass{amsart}
\usepackage{amssymb,amsfonts,amsmath,amsthm,amscd,dsfont,mathrsfs,bbold,pifont}
\usepackage{blkarray}
\usepackage{graphicx,float,psfrag,epsfig,color}
\usepackage{comment}
\usepackage{bbm}
\usepackage{subcaption}
\usepackage{arydshln}
\usepackage{xcolor}
\usepackage{algorithm}
\usepackage[noend]{algpseudocode}
\newcommand{\Hc}{\mathcal{H}}
\newcommand{\Gc}{\mathcal{G}}

\newcommand{\argmax}{\mathrm{argmax}}

\newcommand{\F}{\mathbb{F}} 
 
\newcommand{\mN}{\mathbb{N}}
\newcommand{\mZ}{\mathbb{Z}}

\newcommand{\pp}{\mathbb{P}}
\newcommand{\E}{\mathbb{E}}

\newcommand{\A}{\mathcal{A}}

\newcommand{\1}{\mathbb{1}}

\newcommand{\dist}{\mathrm{dist}}
\newcommand{\Pc}{\mathcal{P}}

\newtheorem{theorem}{Theorem}
\numberwithin{theorem}{section}
\newtheorem{property}[theorem]{Property}

\newtheorem{claim}[theorem]{Claim}
\newtheorem{lemma}[theorem]{Lemma}
\newtheorem{corollary}[theorem]{Corollary}
\newtheorem{definition}[theorem]{Definition}
\newtheorem{remark}[theorem]{Remark}

\newtheorem{conjecture}[theorem]{Conjecture}

\numberwithin{equation}{section}

\title{
Polynomial Freiman-Ruzsa, Reed-Muller codes and Shannon capacity}

\author{E. Abbe}
\address{Mathematics Institute, EPFL}
\curraddr{}
\email{emmanuel.abbe@epfl.ch}
\thanks{}

\author{C. Sandon}
\address{Mathematics Institute, EPFL}
\curraddr{}
\email{colin.sandon@epfl.ch}
\thanks{}

\author{V. Shashkov}
\address{Mathematics Institute, EPFL}
\curraddr{}
\email{vladyslav.shashkov@epfl.ch}
\thanks{}

\author{M. Viazovska}
\address{Mathematics Institute, EPFL}
\curraddr{}
\email{maryna.viazovska@epfl.ch}
\thanks{}


\date{}

\dedicatory{}
\begin{document}

\begin{abstract}
In 1948, Shannon used a probabilistic argument to show the existence of codes achieving a maximal rate defined by the channel capacity. In 1954, Muller and Reed introduced a simple deterministic code construction based on polynomial evaluations, which was conjectured and eventually proven to achieve capacity. Meanwhile, polarization theory emerged as an analytic framework to prove capacity results for a variation of RM codes -- the polar codes. Polarization theory further gave a powerful framework for various other code constructions, but it remained unfulfilled for RM codes. In this paper, we settle the establishment of a polarization theory for RM codes, which implies in particular that RM codes have a vanishing local error below capacity. Our proof puts forward a striking connection with the recent proof of the Polynomial Freiman-Ruzsa conjecture \cite{gowers2023conjecturemarton} and an entropy extraction approach related to \cite{AY18}. It further puts forward a {\it small orbit localization lemma} of potential broader applicability in combinatorial number theory. Finally, a new additive combinatorics conjecture is put forward, with potentially broader applications to coding theory. %{\color{orange} No it would not. In order to prove the strong version we would also need a stronger version of lemma \ref{nearSymSubspace}.} {\color{teal} Further, a new additive combinatorics conjecture is put forward which would imply stronger alignment of random binary vectors with subspaces, potentially aiding with obtaining the vanishing global error result.}{ \color{orange} It is not much stronger polarization. The conjecture would allow us to argue that the probability distribution of $W_r|W_{>r}$ is essentially the same as a probability distribution on a suitable subspace instead of just arguing that it is relatively close to such a probability distribution. However, without a stronger version of lemma \ref{nearSymSubspace} that just gets us better constants.} {\color{olive} Do you think it is a good idea to cite lemma 6.6 in the abstract? I would like to avoid it. I've tried to rewrite it.} {\color{orange} No, that seems too technical. What you have now works.} We expect the latter conjecture to be more directly relevant to coding applications.
\end{abstract}

\maketitle

\newpage
\vspace{-.5cm}
{\footnotesize
\tableofcontents}
\thispagestyle{empty}
\newpage

\setcounter{page}{1}

\section{Coding problem}
Shannon introduced in 1948 the notion of channel capacity \cite{shannon48}, as the largest rate at which messages can be reliably transmitted over a noisy channel. In particular, for the canonical binary symmetric channel which flips every coordinate of a codeword independently with probability $\epsilon$, Shannon's capacity is $1-H(\epsilon)$, where $H$ is the binary entropy function. To show that the capacity is achievable, Shannon used a probabilistic argument, i.e., a code drawn uniformly at random.\footnote{There is also the `worst-case' or `Hamming'  \cite{Hamming50} formulation of the coding problem, where codewords have to be recovered with probability 1 when corrupted by an error rate at most $\epsilon$; there random codes achieve rates up to $1-H(2 \epsilon)$ (or more precisely $1-H(\min(2 \epsilon,1/2))$ as we may have $\epsilon>1/4$),  as codewords there must produce a strict sphere packing of $\epsilon n$-radius spheres (i.e.,  a distance of $2\epsilon n$).}

Obtaining explicit code constructions achieving this limit has since then generated major research activity across electrical engineering, computer science and mathematics. The first decades of coding theory from the 1950s had been dominated by `algebraic codes' \cite{Macwilliams77}, in particular constructions based on polynomial evaluations such as RM codes. While some outstanding constructions were obtained for finite sets of parameters, e.g., the Hamming or Golay codes \cite{Macwilliams77}, proving formal guarantees for algebraic codes in the Shannon setting had seen little progress in the first decades. In the 90s, graph-based codes started to attract major attention, in particular with LDPC codes \cite{Gallager65}, and a proof that expander codes achieve Shannon capacity for the special case of the erasure channel \cite{expander}. The large class of LDPC codes and Turbo codes have also seen major practical developments in telecommunications \cite{Richardson08}. More recently, polar codes \cite{Arikan09} brought a new angle to coding theory, providing a framework, polarization theory, to establish formal guaranties of code achievement of capacity on a symmetric channel. Polar codes are  closely related to RM codes as will be discussed next. In a sense, polar codes are a variant of RM codes with a simplified recursive framework  that allows for simpler proofs and decoders, at the cost of poorer performance metrics such as error rates and distance (the code construction is also less trivial albeit still efficient). Reed-Muller codes were also assumed to have a polarization property, with a conjecture proposed in \cite{Ye19}. However, attempts to establish a polarization result for entropies of individual bits were incomplete. In particular, \cite{Ye19,Hazla21} managed to establish partial-order monotonicity property for the bit entropies. A full monotonicity of bit entropies remained nonetheless out of reach. Block entropies instead (defined below), benefit from a full monotonicity, but lack a polarization result. This will be established here with a connection to the  recently-established Freiman-Ruzsa theorem \cite{gowers2023conjecturemarton}. This connection is achieved in particular with a new \textit{orbit localization lemma} (see Section \ref{}), of potential independent interest in additive combinatorics. 

To introduce the notion of codes achieving capacity, we have to define some key quantities in coding theory. Consider transmitting a message $u \in \mathbb{F}_2^k$ on a noisy channel. To protect the message from the noise, it is embedded in a larger dimension. We define a linear embedding $x:\F_2^k\to\F_2^n$, mapping the message $u$ to the codeword $x=x(u) \in \F_2^n$. The image of $x$ is the linear code that we are going to study. The number $n$ is called the length (or blocklength) of the code, $k$ is the dimension of the code, and the ratio $R:=\frac{k}{n}$ is called the code rate.
The channel model describes the distribution of $\Tilde{x}$ obtained from transmitting $x$. We focus on the central case of the binary symmetric channel (BSC), defined by the following transition probability:
\[P(\Tilde{x}\mid x):=\delta^{|\{i \in [n]\mid  \Tilde{x}_i \neq x_i\}|} (1-\delta)^{|\{i \in [n]\mid  \Tilde{x}_i = x_i\}|}.\]
In simpler terms, the output of the channel is a corruption with i.i.d. Bernoulli noise, i.e., $\Tilde{x}=x+Z,\,\, Z\sim Ber(\delta)^{\F_2^m}$. If $\delta =\Omega_n(1)$ and $n=k$ ($R=1$), we cannot hope to recover $x$ with high probability. However, if $\delta=\Omega_n(1)<1/2$ and $R<1$, one can still hope to recover $x$ with high probability depending on the tradeoff between $R$ and $\delta$. 

%Another technique would be to say that the Boolean vector lives on a linear space. This allows us to incorporate recurrent structures in the code, and this gives us a lot of possibilities for decoding and verification of the result.

We will focus here on linear codes.
%where $U \in \mathbb{F}_2^{nR_n}$ (called a \textbf{\textit{message}}, and $R_n$ is called a \textbf{\textit{code rate}}) needs to be reliably transmitted through the channel.
Let $G\in \mathbb{F}_2^{n \times k}$ be a fixed matrix, often conforming to some recurrent structure along parameters $n$ and $k$. Further, to conform to several code rates, a matrix $G_{full}\in \mathbb{F}_2^{n \times n}$ is defined and $G$ will correspond to a sub-matrix of $G_{full}$ depending on the rate. In this case, $G$ is the code generation matrix and $G_{full}$ is the matrix defined for constructing code generation matrices of flexible rate.
The transmitted codeword will be $Gu$ and the goal is to recover $u$ from $Gu+Z$ with high probability.
%A logical question that one may ask oneself: wouldn't it be efficient to transmit $U$ several times and recover it using majority votes? While a perfectly valid strategy, it is far from being the most efficient. We would need to transmit $U$ $\theta(\log(n))$ times for us to have a vanishing probability that $U$ can be fully recovered. This gives us intuition that we should take $G$ that incorporates $U_i$ into several entries of $GU$ so that we transmit more information with less bits.
In order to minimize the probability of decoding $u$ incorrectly, the optimal algorithm is the  maximum likelihood decoder $\Hat{X}(Y)$ which, in the case of the BSC, outputs the closest codeword to the received word. Assume $U \sim Unif(\mathbb{F}_2^k),\,\, X=GU,\,\, Y=X+Z,\,\, Z \sim Ber(\delta)^{\F_2^m}$.
\begin{itemize}
\item The \textbf{\textit{bit-error probability}} is defined as follows:
\begin{align*}
P_{\mathrm{bit}}:=\max_{i\in[n]} P_{\mathrm{bit},i}:=\max_{i \in [n]} \mathbb{P}(\widehat{X_i}(Y) \ne X_i),
\end{align*}
where $\widehat{X_i}(Y)$ denotes the most likely value of $X_i$ given $Y$, i.e., $\widehat{X_i}(Y)=\argmax_{x_i \in \F_2}\pp(X_i=x_i|Y)$.
\item The \textbf{\textit{block-error probability}} (also called global error) is defined as follows:
\[P_{\mathrm{block}}:=\pp(\Hat{X}(Y)\neq X).\]
\end{itemize}
\begin{definition}[Codes that achieve capacity]
    Consider a family of codes $\left(C_j\right)_{j=1}^{+\infty}$ of length $n_j$ and dimension $k_j$. Let the sequence of rates $\{r_j\}$ defined by $r_j=\frac{k_j}{n_j}$ satisfy $\lim_{j \rightarrow +\infty} r_j=r$. Let $H:[0,\frac{1}{2}]\rightarrow[0,1]$ denote the entropy function.
    \begin{itemize}
        \item For a binary symmetric channel, we say that $\left(C_j\right)_{j=1}^{+\infty}$ achieves the capacity in the weak sense if $\lim_{j \rightarrow +\infty}P_{\mathrm{bit}}(C_j,\,\delta)=0$ for any $\delta \in [0,H^{-1}(1-r))$.
        \item For a binary symmetric channel, we say that $\left(C_j\right)_{j=1}^{+\infty}$ achieves the capacity in the strong sense if $\lim_{j \rightarrow +\infty}P_{\mathrm{block}}(C_j,\,\delta)=0$ for any $\delta \in [0,H^{-1}(1-r))$.
    \end{itemize}
\end{definition}
\begin{remark}
    For any $\delta>H^{-1}(1-r)$, $P_{\mathrm{block}}(C_j,\,\delta)=\Omega_j(1)$. However, remarkably, Shannon has shown that a code drawn uniformly at random achieves capacity in the strong sense with a high probability \cite{shannon48}.  This, in particular, implies that there exist code sequences that achieve capacity in the strong sense. In this paper, of interest is the Reed-Muller code family. We provide an alternative proof to that Reed-Muller code sequences achieve capacity in the weak sense \cite{reeves,abbe2023reedmullercodesvanishingbiterror}, matching the error rate of \cite{abbe2023reedmullercodesvanishingbiterror} and improving the error rate of \cite{reeves}.  
\end{remark}
\section{Reed-Muller codes}
Reed-Muller codes are deterministic codes with a recursive structure. The general notation of Reed-Muller codes is $RM(m,r)$ with parameters $m$ and $r$. Here, $m$ controls the length of the codeword, and $r$ controls the code rate; $n=2^m$, $k={m\choose \le r}$, and $R={m\choose \le r}/2^m$ where $\binom{m}{\leq r}:=\sum_{i=0}^r \binom{m}{i}$. In brief, {\it the code is given by the evaluation vectors of polynomials of bounded degree $r$ on $m$ Boolean variables.}
Here we give a recursive construction of the code.
First, take the $0^n$-codeword as we are building a linear code. Then, as a first column, take $1^n$ as the vector with maximal Hamming distance from $0^n$. As a second column, take a vector that's the furthest away from $0^n$ and $1^n$, such as $(01)^{n/2}$. Complete this to $m+1$ columns to build a code of minimum distance $\frac{n}{2}$ (this is the first order RM code, also called the augmented Hadamard code). This already allows us  to visualize $G_{full}$ for $RM(1,\cdot)$-codes: $G^{(1)}_{full}=\left(\begin{matrix}
  1 & 0 \\ 1 & 1
\end{matrix}\right)$,\,\, $RM(1,0)$ is generated by the first column and $RM(1,1)$ by the first two columns.

The idea of higher order RM codes is to iterate that construction on the support of the previously generated vectors, i.e., the $m+2$-nd column is at distance $n/4$ from all other columns, repeating the pattern $(0 0 0 1)^{n/4}$, and completing these until distance $n/4$ is saturated, which adds $\binom{m}{2}$ more columns. Next one increments $i$, adding $\binom{m}{i}$ columns while only halving the minimum distance.

For $RM(3,\cdot)$, $G_{full}$ is as follows:
$G^{(3)}_{full}=\left(\begin{array}{c:ccc:ccc:c}
  1 & 0 & 0 & 0 & 0 & 0 & 0 & 0  \\
  1 & 1 & 0 & 0 & 0 & 0 & 0 & 0 \\
  1 & 0 & 1 & 0 & 0 & 0 & 0 & 0 \\
  1 & 1 & 1 & 0 & 1 & 0 & 0 & 0 \\
  1 & 0 & 0 & 1 & 0 & 0 & 0 & 0 \\
  1 & 1 & 0 & 1 & 0 & 1 & 0 & 0 \\
  1 & 0 & 1 & 1 & 0 & 0 & 1 & 0 \\
  1 & 1 & 1 & 1 & 1 & 1 & 1 & 1 \\
\end{array}\right)$.

The definition of Reed-Muller codes requires some formal definitions.
\begin{definition}\label{lex}
    Let $S$ be a finite set. Define the following:
    \begin{itemize}
    \item $\F_2^S$ denotes the set of Boolean vectors indexed by the elements of $S$. 
    \item $\binom{S}{r}:=\{S' \subseteq S\mid |S'|=r\}$, $\binom{S}{\lessgtr r}:=\{S' \subseteq S\mid |S'|\lessgtr r\}$.
    \item $\Pc_m:= \F_2[x_1,x_2 \ldots x_m]/(x_i^2=x_i\text{ for }i \in [m])$. 
    \end{itemize}
    In addition, the following maps are introduced.
    \begin{itemize}
        \item $coef:\Pc_m \rightarrow \F_2^{2^{[m]}} - \text{maps Boolean polynomials to their sets of coefficients.}$
        \item $eval:\Pc_m \rightarrow \F_2^{\F_2^m} - \text{maps Boolean polynomials to their value sets.}$
    \end{itemize}
     \end{definition}
\begin{definition}
    For $S_1 \subseteq S_2$, we define the operator $\mathrm{incl}_{S_1,S_2}:\F_2^{S_1} \rightarrow \F_2^{S_2}$  by $(x_s)_{s\in S_1}\mapsto (y_s)_{s\in S_2}$, where $y_s=\begin{cases}x_s,\;s\in S_1\\0,\;s\in S_2\setminus S_1\end{cases}$.

    For $S_1 \subseteq S_2$, define the projection operator $\mathrm{proj}_{S_2, S_1}:\F_2^{S_2} \rightarrow \F_2^{S_1}$ by \[(x_s)_{s \in S_2} \mapsto (x_s)_{s \in S_1}.\]
\end{definition}
Finally, the Reed-Muller code $RM(m,r)$ is defined as follows.
\begin{comment}
\begin{definition}
     Let $m,\,r \in \mathbb{Z}$ satisfy $0 \leq r \leq m$. Define $x_A= \prod_{i \in A}x_i\,\, \text{ for all } A \subseteq [m]$, $x=(x_1,x_2 \ldots x_m) \in \mathbb{F}_2^m$. $f_{RM(m,r)}:\mathbb{F}_2^{\binom{[m]}{\leq r}} \rightarrow \mathbb{F}_2^{\F_2^m}$ is the encoder defined by 
     \[f_{RM(m,r)}(u)=eval\left(\sum_{S \in \binom{[m]}{\leq r}} u_S x_S\right)\]
     $Im(f_{RM(m,r)})=RM(m,r)$ is called a Reed-Muller code.
\end{definition}
\end{comment}
\begin{definition}
     Let $m,\,r \in \mathbb{Z}$ satisfy $0 \leq r \leq m$. Define $x_A= \prod_{i \in A}x_i\,\, \text{ for all } A \subseteq [m]$, $x=(x_1,x_2 \ldots x_m) \in \mathbb{F}_2^m$. $f_{RM}:\mathbb{F}_2^{2^{[m]}} \rightarrow \mathbb{F}_2^{\F_2^m}$ is the encoder defined by 
     \[f_{RM}(u)=eval\left(\sum_{S \in 2^{[m]}} u_S x_S\right).\]
     $Im\left(f_{RM}\circ\mathrm{incl}_{\binom{[m]}{\leq r},2^{[m]}}\right)=f_{RM}\left(\mathbb{F}_2^{\binom{[m]}{\leq r}}\right)=:RM(m,r)$ is called a Reed-Muller code.
    \end{definition}
\begin{remark}[RM code capacity-achieving parameter r.] Note that the rate of $RM(m,r_m)$ is $\frac{\binom{m}{\leq r_m}}{2^m}$, so to attain limit $1-H(\delta)$, $r_m$ must be equal to $\frac{m}{2}+C\sqrt{m}+o_m(\sqrt{m})$ for a specific constant $C \in \mathbb{R}$. The gap between the channel capacity and the code's rate is allowed to be positive, which we exploit by allowing $r$ to be equal to $\frac{m}{2}+(C-\epsilon)\sqrt{m}+o_m(\sqrt{m})$ for an arbitrarily small $\epsilon>0$.
\end{remark}
\begin{comment}
Introduce the following notation: $f(u)=Gu,\,\, G=\left(\begin{matrix}
    1 & 0 \\ 1 & 1
\end{matrix}\right)^{\otimes m},\, m \in \mathbb{N}.$
\begin{remark}\label{perm_of_RM}
$f(u)$ satisfies two properties:
\begin{itemize}
    \item $f(u)$ is a permutation of $f_{RM(m,m)}(u)$. Moreover, $f(coef_{x_1,x_2 \ldots x_m}P(x))=eval_{x_1,x_2 \ldots x_m}P(x)$.
    \item $f \circ f = id$, as $G^2=\left(\left(\begin{matrix} 1 & 0 \\ 1 & 1 \end{matrix}\right)^2\right)^{\otimes m}=I_2^{\otimes m}=I_n$.
\end{itemize}
\end{remark}
\end{comment}
We define the following.
\begin{align*}
    &U\sim Unif(\mathbb{F}_2^{2^{[m]}})\text{ - a coefficient set};\,\,X=f_{RM}(U)\text{ - a codeword of full dimension};\\
    &Y=X+Z\text{ - the observed noisy codeword};\,\,U_r = \mathrm{proj}_{2^{[m]}, \binom{[m]}{r}}(U);\\ &U_{\lessgtr r}=\mathrm{proj}_{2^{[m]}, \binom{[m]}{\lessgtr r}}(U).
\end{align*}
In this paper, we analyze entropies of Reed-Muller message layers. 
\begin{definition}
    Let $A$ be a random variable taking values on a finite set $\A$. Define 
    \[H(A):=- \sum_{a \in \A} \pp(A=a) \log_2 \pp(A=a).\]
    $H(A)$ is called the entropy of $A$.
\end{definition}
The entropy is a convenient measure of randomness to exploit chain rule properties of measuring the randomness on multiple variables, as will be extensively used here for the RM codeword components. 

Starting from this section, we use the following notation:
\[H(A,\,B)=H((A,\,B)), H(A,\,B\mid C,\,D)=H((A,\,B)\mid(C,\,D))\]
for $A,\,B,\,C,\,D$ valued in finite sets.

For a pair of random variables $(A,B)$ valued in $\mathcal{A} \times \mathcal{B}$, the conditional entropy is defined by $H(A\mid B):=H(A,B)-H(B)$. Conditional entropy has the following important property: $H(A\mid B) \leq H(A)$. The expected error probability of the maximum likelihood decoder when guessing $A$ given the observation of $B$ is bounded by $H(A\mid B)$.

\begin{lemma}\label{entbound}
    Consider a random variable $A$ taking values in the set $\mathcal{A}$. Define 
    \[err(A):=1-\max_{a \in \mathcal{A}}\pp_A(a).\]
    Additionally define
    \[err(A\mid B):=\E err(A\mid B=b).\]
    Then, $err(A\mid B)\leq H(A\mid B)$.
\end{lemma}

In our coding setting of RM codes, we are interested in obtaining a small upper bound on the conditional entropy
\[H(U_{\leq r}^{(m)}\mid Y^{(m)}, U_{>r}^{(m)}).\]
This conditional entropy measures the following: we are transmitting a polynomial (of unbounded degree) with random coefficients $U^{(m)}$, and then look at the conditional entropy that the components $U_{\leq r}^{(m)}$ have (which correspond to the $RM(m,r)$ codeword), given the received noisy codeword and the complement components $U_{>r}^{(m)}$ of degree $> r$ that are made available to the decoder. The latter components are made available to the decoder because the RM code freezes these components to 0, and for symmetric channels, this is equivalent w.l.o.g.\ to giving access to $U_{>r}^{(m)}$ to the decoder. Consequently, we aim to decode $U_{\leq r}^{(m)}$ observing $Y^{(m)}, U_{>r}^{(m)}$.

\section{Main result}
In this paper, we provide a polarization theory for RM codes that implies an alternative proof of the weak capacity result. I.e., we show that a {\it monotone entropy extraction phenomenon} takes place for RM codes which implies that RM codes have a vanishing local error at any rate below capacity. The general idea is to show that RM codes will extract the randomness of the code, measuring the latter by using  the Shannon entropy of sequential layers in the code. This approach is similar to the approach developed in \cite{AY18} with the polarization of RM code.

Our proof then relies crucially on the recent Polynomial Freiman-Ruzsa or Marton's conjecture's proof by \cite{gowers2023conjecturemarton}.
The result\footnote{A variant formulation states that for a random variable $X$ on $\mathbb{F}_2^d$ with $d(X,X) \le \log K$, there exists a uniform random variable on a subgroup $\Gc$ of $\mathbb{F}_2^d$ such that $d(X,U_\Gc) \le C \log K$ for a constant $C$ (where $C=6$ is achieved).} of Gowers et al. shows that if $H(X+X')-H(X)$ is small for independent binary vectors $X,X'$ , then $X$ is close to the uniform distribution on a subspace of $\mathbb{F}_2^m$.
This paper shows that this additive combinatorics result is intimately related to the capacity property of Reed-Muller codes when tracking the sequential entropies of RM codes. Namely, how the sequential entropies of the layers of monomials of increasing degree behave as the code dimension grows. The fact that $H(U+U')$ will be bounded away from $H(U)$ due to the entropic Freiman-Ruzsa Theorem will let us show that the subsequent\footnote{One has to actually work with conditional entropies rather than such direct entropies, which also requires us to slightly generalize the result of  \cite{gowers2023conjecturemarton}.} entropies decay as the degree increases, which implies the 'monotone' entropy extraction of the code, allowing for the desired error bounds.

This approach allows us to prove the following result:
\begin{theorem}\label{final_res}
    Consider the binary symmetric channel with error parameter $\delta \in [0,\frac{1}{2})$. Assume that the parameters $m$ and $r_m$ satisfy the relation $\limsup_{m \rightarrow +\infty}\frac{\binom{m}{\le r_m}}{2^m}<1-H(\delta)$, where $0 \leq r_m \leq m$. 
    \begin{enumerate}
    \item{\bf(Layer polarization inequality)}\label{thm: final_res_p2} Let $a_{m,r}=H(U_{\leq r}^{(m)}\mid Y^{(m)}, U_{>r}^{(m)}),\, f_{m,r}=a_{m,r}-a_{m,r-1}=H(U_{r}^{(m)}\mid Y^{(m)}, U_{>r}^{(m)})$. Then, the following block polarization holds:
    \[a_{m+1,r+1}\leq a_{m,r+1}+a_{m,r} - \frac{1}{140}\min\left(f_{m,r+1},\,\binom{m}{r}-f_{m,r+1}\right). \]
    \item{\bf(Layer entropy bound)} Suppose that  $\limsup_{m \rightarrow \infty}\frac{\binom{m}{\leq r_m}}{2^m}= (1-\varepsilon)(1-H(\delta))$ for some $\varepsilon>0$. Then there exists $c_{\varepsilon}>0$ such that $a_{m,r_m}\le 2^m 2^{-2c_{\varepsilon}\sqrt{m}}$ for all sufficiently large~$m$.
    \item{\bf(Weak capacity)} The bit-error probability of the Reed-Muller code sequence $\{RM(m,r_m)\}_{m \in \mN}$ satisfies:
    \[P_{\mathrm{bit}}=2^{-\Omega_m(\sqrt{m})}.\]
    \end{enumerate}
\end{theorem}

\subsection{Additive combinatorics result: orbit localization lemma}\label{sec:add_comb}
The paper establishes a striking connection between the recent proof of the Freiman-Ruzsa conjecture \cite{gowers2023conjecturemarton} and the block entropy polarization. We report here a lemma used in this part, of possible independent interest in additive combinatorics. 

\begin{lemma}
Let $n\in \mN$, $\mathbb{F}$ be a finite field, $\mathcal{T}$ be a set of linear transformations on $\mathbb{F}^n$, and $\mathcal{W}$ be a probability distribution over subspaces of $\mathbb{F}^n$ such that for every $T \in \mathcal{T}$ and every subspace $\Gc_0$ of $\mathbb{F}^n$, the following equality is true: \[\mathbb{P}_{\Gc\sim\mathcal{W}}[\Gc=\Gc_0]=\mathbb{P}_{\Gc\sim\mathcal{W}}[\Gc=T\Gc_0].\]
        Then there must exist a subspace $\Gc^\star$ of $\mathbb{F}^n$ such that $T\Gc^\star=\Gc^\star$ for all $T\in\mathcal{T}$ and
        \[\mathbb{E}_{\Gc\sim\mathcal{W}}[\dist(\Gc,\Gc^\star)]\le\frac{9}{2}\mathbb{E}_{\Gc,\Gc'\sim\mathbb{\mathcal{W}}}[\dist(\Gc,\Gc')]\]
\begin{comment}
        {\color{red} In response to the question about whether $\mathbb{F}$ needs to be a finite field: The core argument still seems to be independent of that, but once there are an infinite number of subspaces of $\mathbb{F}^n$ the requirement that $\mathbb{P}_{\Gc\sim\mathcal{W}}[\Gc=\Gc_0]=\mathbb{P}_{\Gc\sim\mathcal{W}}[\Gc=T\Gc_0]$ is not necessarily meaningful. So, I would need to amend that requirement in some manner. If the subgroup of $GL_n(\mathbb{F})$ induced by $\mathcal{T}$ is finite then I could just require that the support of $\mathcal{W}$ be finite, but this feels like it should still work in some cases where that subgroup is infinite. My first instinct for dealing with that would be to require that $\mathbb{P}_{\Gc\in\mathcal{W}}[\Gc\in S]=\mathbb{P}_{\Gc\in\mathcal{W}}[T\Gc\in S]$ for arbitrary $T\in\mathcal{T}$ and set $S$, but I suspect that that could run into trouble with the Banach-Tarski paradox or comparable weirdness.}
\end{comment}
    \end{lemma}
The small orbit localization lemma allows us to analyze distances arising from invariant probability distributions over subspaces by reducing the problem to the study of invariant subspaces of $\F^n$. The key point is that the family of invariant subspaces $\Gc^*$ form a constrained set, which provides more structure than an arbitrary subspace $\Gc'$. Thus, $\dist(\Gc, \Gc^\star)$ admits stronger control and better bounds than distances between two independently drawn subspaces $\dist(\Gc, \Gc')$. For instance, as we prove later in the paper, in the space of homogeneous $m$-variable Boolean functions of degree $r$ on the field $\F=\F_2$, the only subspaces invariant under all affine linear transformations of $\mathbb{F}_2^m$ are the trivial subspace $\{0\}$ and the entire space. This observation provides a lower bound of $\dist(\Gc, \Gc^\star) \geq \min \{\dim(\Gc), \binom{m}{r}-\dim(\Gc)\}.$

Moreover, the identity $\dist(\Gc, \Gc')=2d(U_\Gc, U_{\Gc'})$ connects distances between subspaces to distances between the associated Boolean random variables with invariant distributions. This relationship enables us to translate structural results about invariant subspaces into quantitative bounds on invariant probability distributions, using tools such as the Freiman–Ruzsa theorem.

\section{Related literature}
It has long been conjectured that RM codes achieve Shannon capacity on symmetric channels,
% We refer to Section \ref{symmetric} for the formal definition of symmetric channels; for now, it is sufficient to consider the BSC, the main case of interest, as symmetric are mixtures of BSCs.
%and perform comparably to random codes on criteria such as the scaling law \cite{Hassani18} or the weight enumerator \cite{Sloane70,Macwilliams77,kasami1970weight,kasami1976weight,Kaufman12,Samorod18}.
%The fact that RM codes have good performance in the Shannon setting, and that they could achieve capacity, has long been observed and conjectured.
%More specifically, the conjecture is as follows.
with its first appearance present shortly after the RM code's definition in the 60s; see \cite{Kudekar17}.
Additional activity supporting the claim took place in the 90s, in particular in 1993 with a talk by
Shu Lin entitled `RM Codes are Not So Bad' \cite{Lin93}. A 1993 paper by Dumer and
Farrell also contains a discussion on the matter \cite{Dumer93}, as does the 1997 paper of Costello and Forney on the `road to channel capacity' \cite{Costello07}. The activity then increased with the emergence of polar codes in 2008 \cite{Arikan09}.
%; Ar\i kan mentioned this as one of the major open problems in the field at ITW Dublin in 2010.
Due to the broad  relevance\footnote{RM codes on binary or non-binary fields have been used for instance in cryptography \cite{Shamir79,BF90,Gasarch04,Yekhanin12}, pseudo-random generators and randomness extractors \cite{DBLP:journals/jcss/Ta-ShmaZS06,bogdanov-viola}, hardness amplification, program testing and interactive/probabilistic proof systems \cite{BFL90,Sha92,ALMSS98}, circuit lower bounds \cite{Razborov}, hardness of approximation \cite{barak2012making,DBLP:conf/focs/BhattacharyyaKSSZ10}, low-degree testing \cite{DBLP:journals/tit/AlonKKLR05,DBLP:journals/siamcomp/KaufmanR06,DBLP:journals/rsa/JutlaPRZ09,DBLP:conf/focs/BhattacharyyaKSSZ10,DBLP:journals/siamcomp/HaramatySS13}, private information retrieval \cite{DBLP:journals/jacm/ChorKGS98,DBLP:journals/jacm/DvirG16,DBLP:journals/jacm/ChorKGS98,BEIMEL2005213,DBLP:conf/focs/BeimelIKR02}, and compressed sensing  \cite{Calderbank2010reed,Calderbank10,Barg15}.} of RM codes in computer science, electrical engineering and mathematics, the activity scattered in a wide line of works \cite{dumer1,dumer2,dumer3,Carlet05,hell,arikan-RM,Arikan09,Arikan2010survey,Kaufman12, Abbe15stoc,Abbe15,Kudekar17,Mondelli14,Saptharishi17,AY18,YA18,Sberlo18,comparingBitMAP,Samorod18,Hazla21,HSS,Lian20,Fathollahi21,ASY21,reeves,Geiselhart21,BhandariHS022,rao2022list}; see also \cite{RM_fnt}.

The approaches varied throughout the last decades:
\begin{itemize}
  \item {\it Weight enumerator:} this approach \cite{Abbe15stoc,Abbe15,Sberlo18,Samorod18,HSS} bounds the global error $P_\mathrm{glo}$ with a bound that handles codewords with the same Hamming weight together. This requires estimating the number of codewords with a given Hamming weight in the code, i.e., the weight enumerator: $A_{m,r}(\alpha)=|\{i \in [2^{mR}]\mid  w_H(X_i) \le \alpha n \}|$, where $\alpha \in [0,1]$ and $w_H$ takes the Hamming weight of its input. The weight enumerator of RM codes has long been studied, in relation to the conjecture and for independent interests, starting with the work of Sloane-Berlekamp for $r=2$ \cite{sloane-RM} and continuing with more recent key improvements based on \cite{Kaufman12} and \cite{Samorod18}.

  \item {\it Area theorem and sharp thresholds:} in this approach from  \cite{Kudekar17}, the local entropy $H(X_i\mid Y_{-i})$ is bounded. By the chain rule of the entropy (i.e., entropy conservation, also called the `area theorem'), if this quantity admits a threshold, it must be located at the capacity. In the case of the erasure channel, this quantity is a monotone Boolean property of the erasure pattern, and thus, results about thresholds for monotone Boolean properties from Friedgut-Kalai \cite{friedgut1996every} apply to give the threshold. Moreover, sharper results about properties with transitive symmetries from  Bourgain-Kalai \cite{Bourgain97} apply to give a $o_n(1/n)$ local error bound, thus implying a vanishing global error from a union bound. The main limitation of this approach is that the monotonicity property is lost when considering  channels that are not specifically the erasure channel (i.e., errors break the monotonicity).

  In \cite{reeves}, this area theorem approach with more specific local error bounds exploiting the nested properties of RM codes is nonetheless used successfully to obtain a local error bound of $O_n(\log\log(n)/\sqrt{\log(n)})$; this gives the first proof of achieving capacity with a vanishing bit-error probability for symmetric channels. This takes place however at a rate too slow to provide useful bounds for the block/global error. Our paper also achieves only a vanishing bit-error probability, with  however an exponential improvement of the rate, namely $2^{-\Omega_n(\sqrt{\log(n)})}$ compared to $O_n(\log\log(n)/\sqrt{\log(n)})$ in  \cite{reeves}. %We further provide in this paper a modified Freiman-Ruzsa inequality conjecture that would improve our rate to $2^{-\Omega(\sqrt{\log(n)} \log\log(n))}$, which in turn could be used to achieve the Shannon capacity result with vanishing block/global error.
  With an additional factor of $\log \log(n)$ in the latter exponent one could use the bit to block error argument from \cite{Abbe23} to obtain a vanishing block error probability; this relates to the modified Freiman-Ruzsa inequality conjecture of Section \ref{new_conj}.

  \item {\it Recursive  methods:} the third approach, related to our paper, exploits the recursive and self-similar structure of RM codes. In particular, RM codes are closely related to polar codes \cite{Arikan09}, with the latter benefiting from martingale arguments in their analysis of the conditional entropies that facilitate the establishment of threshold properties. RM codes have a different recursive structure than polar codes; however, \cite{AY18,Hazla21} show that martingale arguments can still be used for RM codes to show a threshold property, but this requires potentially modifying the RM codes. This prior work focuses on the row by row conditional entropies, establishing the polarization phenomena but obtaining only a partial monotonicity property insufficient to imply that the entropy concentrates on the high-degree rows, i.e., that the original RM codes achieve capacity. In our work, we focus on the layer by layer conditional entropies, which are more easily shown to be monotone. The entropic Freiman-Ruzsa theorem then allows us to reach the polarization phenomenon at the layer scale, which, together with monotonicity, gives the weak capacity result. Self-similar structures of RM codes were also used in \cite{dumer1,dumer2,dumer3,YA18}, but with limited analysis for the capacity conjecture.
  \item {\it Boosting on flower set systems:} Finally the recent paper \cite{Abbe23} settled the conjecture about RM codes achieving a vanishing block/global error down to capacity. The proof relies on obtaining boosted error bounds by exploiting flower set systems of subcodes, i.e., combining large numbers of subcodes' decodings to improve the global decoding. This is a different type of threshold effect that is less focused on being `successive' in the degree of the RM code polynomials, but exploiting more the combination of weakly dependent subcode decodings.
\end{itemize}
Our entropy extraction proof of a vanishing bit-error probability pursues approach 2 related to \cite{AY18}, using successive entropies of RM codes and showing a polarization/threshold phenomenon of the successive layer entropies. The key ingredients to complete this program are: (1) using the Freiman-Ruzsa theorem\cite{gowers2023conjecturemarton} to show that if $H(U^{(m)}_{r}+U'^{(m)}_{r}\mid U^{(m)}_{>r},U'^{(m)}_{>r},Y^{(m)},Y^{'(m)})\approx H(U^{(m)}_{r}\mid U^{(m)}_{>r},Y^{(m)})$ then the probability distribution of $U^{(m)}_{r}\mid U^{(m)}_{>r}, Y^{(m)}$ must be approximately a uniform distribution on a subspace of $\mathbb{F}_2^{\binom{[m]}{r}}$; (2) showing that every subspace of $\mathbb{F}_2^{m\choose r}$ that even approximately satisfies the appropriate symmetries is approximately either the entire space or $\{0\}$; (3) using the previous results to show that the entropies polarize to one of the two extremal values; (4) using the resulting entropy bounds and a list decoding argument to show this in fact gives vanishing bit-error probability.
\section{Set-up for the proof of Theorem 3.1}
\subsection{Compact notation}
The entropy $H(U_{\leq r}^{(m)}\mid Y^{(m)}, U_{>r}^{(m)})$ can be rewritten more compactly. Let $m \in \mN,\,r \in \mZ$ satisfy $0\leq r \leq m$. Let $W:=f^{-1}(Z)$, where $Z\sim Ber(\delta)^{\F_2^{m}}$ is the noise vector, $f(\cdot)=f_{RM}(\cdot)$, $U \sim Unif(\mathbb{F}_2^{2^{[m]}})$. Let $W_r=\mathrm{proj}_{2^{[m]},\binom{[m]}{r}}(W),\, W_{\lessgtr r}=\mathrm{proj}_{2^{[m]},\binom{[m]}{\lessgtr r}}(W)$. The following chain is true:

\begin{equation}\label{simp_eq}
\begin{split}
H(U_r\mid Y,U_{>r})&=H((f^{-1}(Y)+W)_r\mid Y,U_{>r})=H((f^{-1}(Y))_r+W_r\mid Y,U_{>r})\\
&=H(W_r\mid Y,U_{>r})=H(W_r\mid U+W,U_{>r})\\
&=H(W_r\mid W_{>r}, U_{>r}, (U+W)_{\leq r})=H(W_r\mid (U+W)_{\leq r},W_{>r})\\
&=H(W_r\mid W_{>r}).
\end{split}
\end{equation}

\begin{itemize}
    \item The first equality comes from $Y=f(U)+Z$, thus $f^{-1}(Y)=U+f^{-1}(Z)$ and $U=f^{-1}(Y)+f^{-1}(Z)=f^{-1}(Y)+W$.
    \item The third equality comes from removing the conditionally known information from entropy.
    \item The sixth equation comes from the independence of $U_{>r}$ from other random variables.
    \item The seventh equation comes from $U \sim Unif(\mathbb{F}_2^{2^{[m]}})$, which implies that $U+W$ and $W$ are independent.
\end{itemize}
Throughout the work, we equivalently rewrite the entropy of $H(U_r\mid U_{>r},Y)$ as $H(W_r\mid W_{>r})$ and $H(U_{\leq r}\mid U_{>r},Y)$ as $H(W_{\leq r}\mid W_{>r})$, simplifying the notation. Note that in cases where the parameter $m$ changes, we use an alternative notation $W^{(m)}_r$ to avoid confusion.
\subsection{Ruzsa distance and symmetries}
Throughout the work, we need to compare entropies of sums of random variables to their individual entropies. This comes up from the recurrent structure of RM codes (Plotkin construction), and the fact that we look at layers of RM codes of consecutive degree. More precisely, one has:
\begin{equation}\label{eq_almost_rec}
\begin{split}
  &H(W^{(m+1)}_{\leq r}\mid W^{(m+1)}_{>r})=H(W^{(m)}_{\leq r},W'^{(m)}_{\leq r-1}\mid W^{(m)}_{>r},W'^{(m)}_{>r},W^{(m)}_{r}+W'^{(m)}_{r})=\\
  &=2H(W^{(m)}_{\leq r}\mid W^{(m)}_{>r})-H(W^{(m)}_{r}+W'^{(m)}_{r}\mid W^{(m)}_{>r},W'^{(m)}_{>r}).
\end{split}
\end{equation}

In order to deal with entropies of type $H(X+X')-H(X)$, the notion of Ruzsa distance is used.

\begin{definition}\label{def:ruzsa}
Let $X,Y$ be random variables taking values in an abelian group $\mathrm{G}$, and $X', Y'$ be variables that have the same probability distributions as $X$ and $Y$ but are independent of each other. Define
    \[d(X,Y):=H(X'-Y')-\frac{1}{2}(H(X')+H(Y')).\]
    $d(X,Y)$ is called the Ruzsa distance of $X$ and $Y$.
\end{definition}

The Ruzsa distance is not formally a distance; $d(X,X)>0$ for any $X$ valued in $\mathbb{F}_2^d$ that is not equivalent to a uniform random variable on a coset of some subspace $\Hc \subseteq \mathbb{F}_2^d$. However, the Ruzsa distance satisfies the triangle inequality and is nonnegative.

\begin{property}
    Let $X, Y, Z$ be random variables taking values on an abelian group $(\mathrm{G},+)$. The following relation is true:
    \begin{enumerate}
        \item $d(X,Y)\geq\frac{1}{2}|H(X)-H(Y)|,$
        \item $d(X,Y)\leq d(X,Z)+d(Z,Y).$
    \end{enumerate}
\end{property}
We define the conditional Ruzsa distance in order to deal with conditional entropies.
\begin{definition}
     Let $X,Y$ be $\mathrm{G}$-valued random variables with $(\mathrm{G},+)$ being an abelian group, $A,B$ be random variables; and $(X',A')$ and $(Y',B')$ be copies of $(X,A)$ and $(Y,B)$ that are independent of each other. Define
    \[d(X\mid A,Y\mid B):=H(X'-Y'\mid A',B')-\frac{1}{2}(H(X'\mid A')+H(Y'\mid B')).\]
    $d(X\mid A, Y\mid B)$ is called the conditional Ruzsa distance of $X$ conditioned on $A$ and $Y$ conditioned on $B$.
\end{definition}
The object of interest is $d(W_r\mid W_{>r},\, W_r\mid W_{>r})$, as it appears in \ref{eq_almost_rec}. 
\subsection{Freiman-Ruzsa inequality for conditional Ruzsa distance}

\begin{theorem}[Entropic Freiman-Ruzsa Theorem \cite{gowers2023conjecturemarton}]
    Let $k \in \mN$. For any $\mathbb{F}_2^k$-valued random variables $X,Y$, 
    \[\exists \text{ subspace } \Gc \subseteq \mathbb{F}_2^k: d(X,U_\Gc)\leq 6d(X,Y)\,\,\]
where $U_{\Gc}$ is the uniform distribution on $\Gc$.
\end{theorem}
The Entropic Freiman-Ruzsa Theorem does not allow us to directly work with the conditional entropy. This is due to the fact that $\Gc$ depends on $X$, as such an averaging approach fails. However, the following corollary overcomes this limitation. Define $d(X\mid Y,Z):=H(X+Z\mid Y)-\frac{1}{2}(H(X\mid Y)+H(Z))=:d(Z,X\mid Y).$

\begin{corollary}[Conditional Entropic Freiman-Ruzsa Theorem]\label{CEFR}
     Let $k \in \mN$. For $\mathbb{F}_2^k$-valued random variables $X,\,Y$, arbitrarily valued random variables $A,\, B$, there exists a subspace $\Gc$ of $\mathbb{F}_2^k$ such that $d(Y\mid B,U_\Gc) \leq 7d(X\mid A,Y\mid B)$ where $U_\Gc$ is a uniform random variable on $\Gc$.
\end{corollary}

\begin{proof}
    Consider $d(X\mid A, Y\mid B)$. As $d(X\mid A, Y\mid B)=\E_w d(X\mid A=w, Y\mid B)$, there exists a $w$ such that $d(X\mid A=w, Y\mid B) \leq d(X\mid A, Y\mid B)$. 

    Take $\Gc^*=\text{argmin}_\Gc d(X\mid A=w, U_\Gc)$. Since there are finitely many subspaces of $\mathbb{F}_2^k$, $\Gc^*$ exists. As such, $d(X\mid A=w, U_{\Gc^*}) \leq 6d(X\mid A=w, Y\mid B=w')$ for any $w'$. Taking an expectation over $w'$, we get $d(X\mid A=w, U_{\Gc^*}) \leq 6d(X\mid A=w, Y\mid B)$. Finally, using the triangle inequality,
    \begin{align*}
    &d(Y\mid B, U_{\Gc^*}) \leq d(X\mid A=w, U_{\Gc^*})+d(X\mid A=w, Y\mid B)\\
    &\leq 7d(X\mid A=w, Y\mid B) \leq 7d(X\mid A, Y\mid B).
    \end{align*}
    Here, the triangle inequality is used as follows: $d(Y\mid B, U_{\Gc^*}) \leq d(X\mid A=w, U_{\Gc^*})+d(X\mid A=w, Y\mid B)$, which works due to an expectation $\E_{w'}$ put on both sides of $d(Y\mid B=w', U_{\Gc^*}) \leq d(X\mid A=w, U_{\Gc^*})+d(X\mid A=w, Y\mid B=w')$.   
\end{proof}

\begin{comment}
To simplify the computations, leveraging $f \circ f =id$, consider a new ordering of coefficients $U^{new}$ compared to the old ordering $U^{old}$ such that $X=f(U^{new})=f_{RM(m,m)}(U^{old})$, where $f(x)=\left(\begin{array}{cc}
    1 & 0 \\
    1 & 1
\end{array}\right)^{\otimes m}x$ , This, in set indexing notation, translates into
\[U^{new}_i:=U^{old}_{\phi({\{j \in [m]\mid (\mathrm{bin}_m(i-1))_j=1\})}},\]
where $\mathrm{bin}_m(i)$ is a binary representation of $i$ satisfying $\mathrm{bin}_m(2^k)=0^k10^{m-k-1}$ for all $0\leq k\leq m-1$ and $\phi:2^{[m]} \rightarrow [n]$ is a numbering of sets preserving total order \ref{lex}.
From here on, we refer to $U^{new}$ as $U$, and to $U^{old}_A$ as $U_A$.
{\color{red} What does $binom_m(i-1))$ mean?} {\color{green} Explained} 

We introduce the following notation: 
\begin{align*}
    &X=f(U)\text{ - a true codeword};\,\,Y=f(U)+Z\text{ - observed noisy codeword};\\
    &U\sim Unif(\mathbb{F}_2^n)\text{ - coefficient set};\,\,U_r = \{U_A\}_{|A|=r};\,\, U_{\lessgtr r}=\{U_A\}_{|A|\lessgtr r};\\
    &Z\sim Ber(\delta)^n\text{ - noise vector};\,\, W=f(Z)\text{ - transformed noise vector from \eqref{simp_eq}};\\
    &W_r=\{W_i\mid i \in [n],\,\, \sum_j \mathrm{bin}_m(i-1)_j = r\};\\
    &W_{\lessgtr r}=\{W_i\mid i \in [n],\,\, \sum_j \mathrm{bin}_m(i-1)_j \lessgtr r\}.
\end{align*}
\end{comment}

\section{Proof of Theorem 3.1}
This section is set up in the following way. The first subsection focuses on proving the first point of Theorem 3.1. In this subsection, the first subsubsection proves the affine invariance of Reed-Muller-associated Ruzsa distance $d(W_r \mid W_{>r}, U_\Gc)$. The second subsubsection provides a lower bound for $\max_{\pi}d(U_\Gc, \pi(U_\Gc))$ over affine transformations $\pi$, culminating in Corollary \ref{cor:dist_bound}. The third subsubsection establishes the recurrent layer entropy inequality using the Freiman-Ruzsa inequality and Corollary \ref{cor:dist_bound}. Finally, the fourth subsubsection establishes the layer polarization inequality based on the recurrent layer entropy inequality. 

The second subsection focuses on proving the second point of Theorem 3.1. It establishes the upper bound on layer entropy based on the layer polarization inequality. 

The third subsection builds the list of decoding candidates based on the value of the layer entropy and utilizes the list to bound the bit error probability. Note that every subsection here only uses the result of the previous subsection, and thus the theorem statements can be generalized to other codes with similar polarization properties.
    \subsection{Proof of Theorem 3.1 (1)}
    In this subsection, we will establish lower bounds on $d(W^{(m)}_{r}\mid W^{(m)}_{>r},W'^{(m)}_{r}\mid W'^{(m)}_{>r})$ so that we can prove and use equation \ref{eq_almost_rec} to bound $H(W^{(m+1)}_{\leq r}\mid W^{(m+1)}_{>r})$. Here, $m \in \mN, r \in \mZ$ satisfy $0 \leq r \leq m$. Our first step toward doing that will be to use the Freiman-Ruzsa Theorem to argue that if this distance is small then $W^{(m)}_{r}\mid W^{(m)}_{>r}$ must be close to a uniform distribution on a subspace of $\mathbb{F}_2^{{[m]\choose r}}$.

    Note that by the conditional entropic Freiman-Ruzsa theorem,
    \[\exists \text{ subspace } \Gc \subseteq \mathbb{F}_2^{\binom{[m]}{r}}: d(U_\Gc, W_r\mid W_{>r}) \leq 7d(W_r\mid W_{>r}, W'_r\mid W'_{>r}).\]
    First, we prove the permutation invariance of $d(U_\Gc, W_r\mid W_{>r})$. Next, we study the properties of the invariance group. Finally, we use these properties to derive the recurrence bound.
    \subsubsection{Permutation invariance}\label{Sec5.2.1}
    The Ruzsa distance is a useful notion to exploit the symmetrical structure of Reed-Muller codes. Note that $d(X,Y)=d(\pi(X),\pi(Y))$ when $\pi:\mathbb{F}_2^d\rightarrow\mathbb{F}_2^d$ is an isomorphism. Reed-Muller codes are symmetric with respect to a large family of symmetries, called "affine transformations". 
\begin{definition}
     Let $m \in \mathbb{N}$. Let $(\overline{x}, u_1) \in \mathbb{F}_2^{\F_2^m} \times \mathbb{F}_2^{2^{[m]}}$. Also, let $g_{A}(x)=Ax\,\, \text{ for all } A \in GL_m(\mathbb{F}_2),\,\,x \in \mathbb{F}_2^m$. Let $h=coef^{-1}(u_1),\, h^\star=eval^{-1}(\overline{x})$. Then, we define the following:
    \[g_{A}^{coef}(u_1):=coef(h(Ax)),\,\, g_{A}^{eval}(\overline{x}):=eval(h^\star(Ax)).\]
\end{definition}
\begin{remark}
    Let $m \in \mathbb{N},\, r \in \mZ$ satisfy $0 \leq r \leq m$ and $A \in GL_m(\mathbb{F}_2),\, \overline{x}\in \mathbb{F}_2^{\F_2^m}$. If $\overline{x} \in RM(m,r)$, then $g_{A}^{eval}(\overline{x}) \in RM(m,r)$.
\end{remark}
In order to exploit properties of $g^{coef}_{A}$, we define the following operator.

This notion of symmetry has some nice properties which we exploit. 
\begin{property}\label{prop:automorph}    
    Let $m\in \mN,\, r \in \mZ$ satisfy $0\leq r \leq m$. $g_{A}^{coef}$ and $g_{A}^{eval}$ are isomorphisms of the vector spaces $\mathbb{F}_2^{2^{[m]}}\rightarrow \mathbb{F}_2^{2^{[m]}}$ and $\mathbb{F}_2^{\F_2^m}\rightarrow \mathbb{F}_2^{\F_2^m}$. Also, $g^{coef}_{A,\leq r}=\mathrm{proj}_{\F_2^{2^{[m]}} \rightarrow \F_2^{\binom{m}{\leq r}}}\circ g_{A}^{coef}\circ \mathrm{incl}_{\F_2^{\binom{m}{\leq r}}\rightarrow \F_2^{2^{[m]}}}$ is an invertible linear map and $g_{A}^{eval}$ is a permutation. 
\end{property}
\begin{comment}
\begin{proof}
     First, we prove that $g_{A,b}^{eval}$ is a permutation of coordinates. Note that $Ax+b$ defines a permutation $\pi \in S_{\mathbb{F}_2^m}$ of binary vectors. The following is thus true for any Boolean polynomial $P$: 
    \[eval(P(Ax+b))=\pi(eval(P(x))).\]
    Now, we prove $g^{coef}_{A,b,\leq r}$ is an invertible linear map. First, note that $g_{A,b}^{coef}$ is a homomorphism as:
    \begin{align*}
    &g_{A,b}^{coef}(u+v)=coef (P+Q)(Ax+b)\\
    &=coef(P(Ax+b))+coef(Q(Ax+b))=g_{A,b}^{coef}(u)+g_{A,b}^{coef}(v)
    \end{align*}
   Here, $P=coef^{-1}(u),\,\, Q=coef^{-1}(v).$
    $g^{coef}_{A,b,\leq r}$ is contained in $V_{m,r}$, because if $P(x)$ is a degree $r$ polynomial, $P(Ax+b)$ is also a degree $r$ polynomial. Moreover, $g_{A,b}^{-1}(x)=A^{-1}x+A^{-1}b$, as such $g_{A,b}$ is bijective. So, $g_{A,b}^{coef}|_{V_{m,r}}$ is bijective and its input space is equal to its output space, thus $g_{A,b}^{coef}|_{V_{m,r}}$ is an automorphism.
\end{proof}
\end{comment}
\begin{corollary}\label{corr:automorph}
Let $m \in \mN,\, r \in \mathbb{Z}$ satisfy $0 \leq r \leq m$. For all $A \in GL_m(\F_2)$, the following map $\mathrm{proj}_{\F_2^{2^{[m]}} \rightarrow \F_2^{\binom{m}{\geq r+1}}}\circ g_{A}^{coef}=:\pi_2:\F_2^{2^{[m]}} \rightarrow \F_2^{\binom{[m]}{\geq r+1}}$ is a surjective linear map satisfying $\pi_2=\pi_2\circ\mathrm{incl}_{\F_2^{\binom{m}{\geq r+1}} \rightarrow \F_2^{2^{[m]}}} \circ \mathrm{proj}_{\F_2^{2^{[m]}} \rightarrow \F_2^{\binom{m}{\geq r+1}}}$. Moreover, $g^{coef}_{A,r}:=\mathrm{proj}_{\F_2^{2^{[m]}} \rightarrow \F_2^{\binom{m}{r}}}\circ g_{A}^{coef}\circ \mathrm{incl}_{\F_2^{\binom{m}{r}}\rightarrow \F_2^{2^{[m]}}}$ is an invertible linear map.
\end{corollary}
\begin{comment}
\begin{proof}
The projection of $g_{A,b}^{coef}(u_1,\{0\}^{\binom{[m]}{\geq r+1}})$ onto $\mathbb{F}_2^{\binom{[m]}{\geq r+1}}$ is $\{0\}^{\binom{[m]}{\geq r+1}}$ for all $u_1$, so there must exist some homomorphism $\pi_2$ such that $g_{A,b}^{coef}(u_1, u_2)$'s projection on $\mathbb{F}_2^{\binom{[m]}{\geq r+1}}$ is $\pi_2(u_2)$ for all $u_1$ and $u_2$. If there was any value in $ \mathbb{F}_2^{\binom{[m]}{\geq r+1}}$ that $\pi_2(u_2)$ never equaled then $g_{A,b}^{coef}(u_1, u_2)$ would never equal anything that projects to it, but this cannot be the case because $g_{A,b}^{coef}$ is an automorphism. So, $\pi_2$ is also an automorphism. 

Now consider $\pi_2\Big(u_3, \{0\}^{\binom{[m]}{\geq r+2}}\Big)$. Note that $u_3$ corresponds to the coefficients of degree $r+1$. The output belongs to $\mathbb{F}_2^{\binom{[m]}{r+1}} \times \{0\}^{\binom{[m]}{\geq r+2}}$ due to Proposition \ref{prop:automorph}, namely - the image of $g_{A,b}^{coef}\Big(\cdot, \{0\}^{\binom{[m]}{\geq r+2}}\Big)$ belongs to $\mathbb{F}_2^{\binom{[m]}{\leq r+1}} \times \{0\}^{\binom{[m]}{\geq r+2}}$, so the projection of $g_{A,b}^{coef}\Big(\{0\}^{\binom{[m]}{\leq r}}, u_3, \{0\}^{\binom{[m]}{\geq r+2}}\Big)$, which is $\pi_2\Big(u_3, \{0\}^{\binom{[m]}{\geq r+2}}\Big)$, belongs to $\mathbb{F}_2^{\binom{[m]}{r+1}} \times \{0\}^{\binom{[m]}{\geq r+2}}$. We have previously proven that $\pi_2$ is an isomorphism, so $\pi_2\Big(\cdot, \{0\}^{\binom{[m]}{\geq r+2}}\Big)$ is also an isomorphism.
\end{proof}
\end{comment}
\begin{definition}
    Let $m \in \mN,\, r \in \mZ$ satisfy $0 \leq r \leq m$. Let $\pi$ be an isomorphism on $\mathbb{F}_2^{\binom{[m]}{r}} \rightarrow \mathbb{F}_2^{\binom{[m]}{r}}$, and $\Gc$ be a subspace of $\F_2^{\binom{[m]}{r}}$. Define $\pi(\Gc):=\{\pi(g)\mid  g \in \Gc\}.$
\end{definition}
Note: $\pi(U_\Gc)=U_{\pi(\Gc)}$.\\

\begin{lemma}\label{perm}
    Let $m \in \mN,\, r \in \mZ$ satisfy $0 \leq r \leq m$ and let $\Gc$ be a subspace of $\F_2^{\binom{[m]}{r}}$. For any $A \in GL_m(\F_2)$, %$d(U_\Gc,W_r\mid W_{>r})=d(g^{coef}_{A,r}(U_\Gc),W_r\mid W_{>r})$. 
    $d(U_\Gc,W_r\mid W_{>r})=d(g^{coef}_{A,r}(U_\Gc),W_r\mid W_{>r})$ .
    
\end{lemma}

    \begin{proof}
    %Finally, we wish to prove $d(U_\Gc,W_r\mid W_{>r})=d(g^{coef}_{A,r}(U_\Gc),W_r\mid W_{>r})$ to exploit the properties of linear transformations. 
    Note that \scalebox{0.94}{$d(U_\Gc,W_r\mid W_{>r})=H(U_\Gc+W_r \mid W_{>r})-\frac{1}{2} \left(H(U_\Gc)+H(W_r \mid W_{>r}) \right)$}, \scalebox{0.91}{$d(g^{coef}_{A,r}(U_\Gc),W_r\mid W_{>r})=H(g^{coef}_{A,r}(U_\Gc)+W_r \mid W_{>r})-\frac{1}{2} \left(H(g^{coef}_{A,r}(U_\Gc))+H(W_r \mid W_{>r}) \right)$}.\\ The equality $H(U_\Gc)=H(g^{coef}_{A,r}(U_\Gc))$ follows from $g^{coef}_{A,r}(\cdot)$ being a bijection, so we only need to prove $H(g^{coef}_{A,r}(U_\Gc)+W_r \mid W_{>r})=H(U_\Gc+W_r \mid W_{>r}).$
    
    Note the following:
    \begin{equation}\label{weighted_noise_invariance}
    \begin{split}
    &\pp(W=w)=\pp(Z=f(w))\\
    &=\pp(Z=g^{eval}_{A}(f(w)))=\pp(W=f^{-1}(g^{eval}_{A}(f(w))))=\pp(W=g^{coef}_{A}(w)).
    \end{split}
    \end{equation}
\begin{itemize}
    \item The first equality follows from the definition of $W=f^{-1}(Z)$.
    \item The second equality follows from the fact that $g_{A}^{eval}$ is a permutation and thus does not change weight. $\pp(Z=w)$ is a function of Hamming weight of $w$, as such $\pp(Z=w)$ is permutation-invariant.
    \item The third equality follows from $W=f^{-1}(Z)$.
    \item The fourth equality follows from the following argument. Let $coef(P(x))=w$, then $f(w)=eval(P(x))$. Consequently, 
    \begin{align*}
    g^{eval}_{A}(f(w))=g^{eval}_{A}(eval(P(x)))=eval(P(Ax)).
    \end{align*}
    Finally, due to $f^{-1}(eval(P(x)))=coef(P(x))$, we conclude that 
    \begin{align*}
        f^{-1}(g^{eval}_{A}(f(w)))=f^{-1}(eval(P(Ax)))=coef(P(Ax))=g^{coef}_{A}(w).
    \end{align*}
\end{itemize} 

For each $A \in GL_m(\mathbb{F}_2)$ the following linear maps are defined:
%&\mathrm{proj}_{\F_2^{2^{[m]}} \rightarrow\F_2^{\binom{[m]}{\leq r}}}\circ g_{A,b}^{coef}\circ \mathrm{incl}_{\F_2^{\binom{[m]}{\leq r}} \rightarrow \F_2^{2^{[m]}}}=:\pi_{0, \leq r}:\mathbb{F}_2^{\binom{[m]}{\leq r}}\rightarrow \mathbb{F}_2^{\binom{[m]}{\leq r}},\\
\begin{align*}
    &\mathrm{proj}_{\F_2^{2^{[m]}} \rightarrow\F_2^{\binom{[m]}{\leq r}}}\circ g_{A}^{coef}\circ \mathrm{incl}_{\mathbb{F}_2^{\binom{[m]}{> r}} \rightarrow \F_2^{2^{[m]}}}=:\pi_{0, >r}:\mathbb{F}_2^{\binom{[m]}{> r}}\rightarrow \mathbb{F}_2^{\binom{[m]}{\leq r}},\\
    &\mathrm{proj}_{\F_2^{2^{[m]}} \rightarrow \mathbb{F}_2^{\binom{[m]}{r}}}\circ g_{A}^{coef}\circ \mathrm{incl}_{\mathbb{F}_2^{\binom{[m]}{>r}} \rightarrow \F_2^{2^{[m]}}}=:\pi_{>r}:\mathbb{F}_2^{\binom{[m]}{>r}}\rightarrow \mathbb{F}_2^{\binom{[m]}{r}},\\
    &\mathrm{proj}_{\F_2^{2^{[m]}} \rightarrow \mathbb{F}_2^{\binom{[m]}{>r}}}\circ g_{A}^{coef}\circ \mathrm{incl}_{\mathbb{F}_2^{\binom{[m]}{>r}} \rightarrow \F_2^{2^{[m]}}}=:g^{coef}_{A,>r}:\mathbb{F}_2^{\binom{[m]}{>r}}\rightarrow \mathbb{F}_2^{\binom{[m]}{>r}}.
\end{align*}
such that for every $W \in \mathbb{F}_2^{2^{[m]}}$ and $W':=g_{A}^{coef}(W)$, the following conditions hold:
%\begin{align*}
%    &g_{A,b}^{coef}(x)=(\pi_{0,\leq r}(x_{\leq r})+\pi_{0,>r}(x_{>r}),\pi_{1,>r}(x_{>r})),\\
%    &\pi_{0,\leq r}:\mathbb{F}_2^{\binom{m}{\leq r}} \rightarrow \mathbb{F}_2^{\binom{m}{\leq r}},\,\, \pi_{0,>r}:\mathbb{F}_2^{\binom{m}{> r}} \rightarrow \mathbb{F}_2^{\binom{m}{\leq r}},\,\,\pi_{1,>r}:\mathbb{F}_2^{\binom{m}{> r}} \rightarrow \mathbb{F}_2^{\binom{m}{> r}},\\
%    &\pi_{1,>r-1}(x_{>r-1})=(\pi(x_r)+\pi_{>r}(x_{>r}),\pi_{1,>r}(x_{>r})),\\
%    &\pi:\mathbb{F}_2^{\binom{m}{r}}\rightarrow \mathbb{F}_2^{\binom{m}{r}} \text{ - automorphism from the Lemma statement},\,\, \pi_{>r}:\mathbb{F}_2^{\binom{m}{> r}} \rightarrow \mathbb{F}_2^{\binom{m}{r}}.
%\end{align*}
\begin{align*}
    &W'_{\le r}=g^{coef}_{A,\leq r}(W_{\leq r})+\pi_{0,>r}(W_{>r}),\\
    &W'_r=g^{coef}_{A,r}(W_r)+\pi_{>r}(W_{>r}),\\
    &W'_{>r}=g^{coef}_{A,>r}(W_{>r}).
\end{align*}
The second and third relations do not depend on $W_{<r}$ and $W_{\leq r}$ respectively due to Proposition \ref{prop:automorph} and Corollary \ref{corr:automorph}.

By Corollary \ref{corr:automorph} , $g^{coef}_{A,>r}$ is an isomorphism, and the same corollary implies that $g^{coef}_{A,r}$ is an isomorphism as well.  By Property \ref{prop:automorph}, $g^{coef}_{A,\leq r}$ is an isomorphism. Note the following properties:
\begin{align*}
    &\pp(W_{>r}=w_{>r})=\sum_{w_{\leq r}} \pp(W_{\leq r}=w_{\leq r}, W_{>r}=w_{>r})=\sum_{w_{\leq r}} \pp(W=g^{coef}_{A}(w_{\leq r},w_{>r}))\\
    &=\sum_{w_{\leq r}} \pp(W_{\leq r}=g^{coef}_{A,\leq r}(w_{\leq r})+\pi_{0,>r}(w_{>r}), W_{>r}=g^{coef}_{A,>r}(w_{>r}))\\
    &=\pp(W_{>r}=g^{coef}_{A,>r}(w_{>r}))
\end{align*}
The second equality follows from \eqref{weighted_noise_invariance} and the last equality follows from the bijectivity of $g^{coef}_{A,\leq r}$, specifically from the fact that $g^{coef}_{A,\leq r}(w_{\leq r})$ attains all the values in $\mathbb{F}_2^{\binom{[m]}{\leq r}}$. Moreover, one can similarly show the following relation:
\[\pp(W_r=w_r\mid W_{>r}=w_{>r})=\pp(W_r=g^{coef}_{A,r}(w_r)+\pi_{>r}(w_{>r})\mid W_{>r}=g^{coef}_{A,>r}(w_{>r})).\]
\begin{comment}
{\color{red} I do not think you ever defined $\pi_{0,\le r}$, $\pi_{0,> r}$, or $\pi_{1,> r}$ and they do not look like they are supposed to parallel $W^{(m)}_{0,\le r}$ and the like.} {\color{green} Thanks, added explanation} {\color{red} That is not really a definition; you just listed some properties that they have vwithout saying what they are. Did you define the $(x_{\le r},x_{>r})$ format for expressing an element of $\mathbb{F}_2^n$ somewhere?}{\color{green} Modified the Corollary 5.19 and Property 5.18, as well as added an alternative explanation. Added Remark 5.20 to relate $g_{A,b}^{coef}$ to $g_{A,b}^{coef'}$} 
\end{comment}
    %Note that $H(g^{coef}_{A,r}(U_\Gc))=\log_2|g^{coef}_{A,r}(\Gc)|=\log_2|G|=H(U_\Gc)$. Thus, it would alternately suffice to prove 
    %\[H(U_\Gc+W_r\mid W_{>r})=H(g^{coef}_{A,r}(U_\Gc)+W_r\mid W_{>r}).\] 
    %The proof is as follows:
    We proceed as follows:
    \begin{align*}
    &-H(U_\Gc+W_r\mid W_{>r})\\    &=\sum_{w_{\Gc,r},w_{>r}}\pp(U_\Gc+W_r=w_{\Gc,r},W_{>r}=w_{>r})\log_2\pp(U_\Gc+W_r=w_{\Gc,r}\mid W_{>r}=w_{>r}).
    \end{align*}
    We transform the probability term in the sum as follows:
    \begin{align*}   
    &\pp(U_\Gc+W_r=w_{G,r},W_{>r}=w_{>r})=\sum_{u_\Gc}\pp(U_\Gc=u_\Gc, W_r=w_{\Gc,r}+u_\Gc,W_{>r}=w_{>r})\\
    &=\sum_{u_\Gc}\Bigg(\pp(g^{coef}_{A,r}(U_\Gc)=g^{coef}_{A,r}(u_\Gc))\\
    &\cdot\pp(W_r=g^{coef}_{A,r}(w_{\Gc,r})+g^{coef}_{A,r}(u_\Gc)+\pi_{>r}(w_{>r}),W_{>r}=g^{coef}_{A,>r}(w_{>r}))\Bigg)\\
    &=\pp(g^{coef}_{A,r}(U_\Gc)+W_r=g^{coef}_{A,r}(w_{\Gc,r})+\pi_{>r}(w_{>r}),W_{>r}=g^{coef}_{A,>r}(w_{>r})).
    \end{align*}
    Similarly,
    \begin{align*}
    &\pp(U_\Gc+W_r=w_{G,r}\mid W_{>r}=w_{>r})\\
    &=\pp(g^{coef}_{A,r}(U_\Gc)+W_r=g^{coef}_{A,r}(w_{\Gc,r})+\pi_{>r}(w_{>r})\mid W_{>r}=g^{coef}_{A,>r}(w_{>r})).
    \end{align*}
    Note that $g^{coef}_{A,r}$ is a bijective mapping $\mathbb{F}_2^{\binom{[m]}{r}} \rightarrow \mathbb{F}_2^{\binom{[m]}{r}}$, which implies
    \begin{align*}
    &\sum_{w_{\Gc,r}}\Bigg(\pp(g^{coef}_{A,r}(U_\Gc)+W_r=g^{coef}_{A,r}(w_{\Gc,r})+\pi_{>r}(w_{>r}),W_{>r}=g^{coef}_{A,>r}(w_{>r})) \\
    & \cdot \log_2 \pp(g^{coef}_{A,r}(U_\Gc)+W_r=g^{coef}_{A,r}(w_{\Gc,r})+\pi_{>r}(w_{>r})\mid W_{>r}=g^{coef}_{A,>r}(w_{>r}))\Bigg) \\
    &= \sum_{w'_{\Gc,r}} \Bigg(\pp(g^{coef}_{A,r}(U_\Gc)+W_r=w'_{\Gc,r},W_{>r}=g^{coef}_{A,>r}(w_{>r}))\\
    &\cdot \log_2 \pp(g^{coef}_{A,r}(U_\Gc)+W_r=w'_{\Gc,r}\mid W_{>r}=g^{coef}_{A,>r}(w_{>r}))\Bigg),
    \end{align*}
    for any $w_{>r}$ where $w'_{\Gc,r}=g^{coef}_{A,r}(w_{\Gc,r})+\pi_{>r}(w_{>r})$ goes through all the values of $\mathbb{F}_2^{\binom{[m]}{r}}$. Analogically, as $g^{coef}_{A,>r}$ is also a bijection, one may note the following:

    \begin{align*}
    &\sum_{w_{>r}}\Bigg( \pp(g^{coef}_{A,r}(U_\Gc)+W_r=w'_{\Gc,r},W_{>r}=g^{coef}_{A,>r}(w_{>r}))\\
    &\cdot\log_2 \pp(g^{coef}_{A,r}(U_\Gc)+W_r=w'_{\Gc,r}\mid W_{>r}=g^{coef}_{A,>r}(w_{>r}))\Bigg)\\
    &=\sum_{w'_{>r}}\Bigg(\pp(g^{coef}_{A,r}(U_\Gc)+W_r=w'_{\Gc,r},W_{>r}=w'_{>r})\\
    &\cdot\log_2 \pp(g^{coef}_{A,r}(U_\Gc)+W_r=w'_{\Gc,r}\mid W_{>r}=w'_{>r})\Bigg),
    \end{align*}
    for any $w'_{\Gc,r}$ where $w'_{> r}=g^{coef}_{A,>r}(w_{>r})$ goes through all the values of $\mathbb{F}_2^{\binom{[m]}{>r}}$.
\begin{comment}
    {\color{red} How does it imply that?}{\color{green} Would it be sufficient to say that $\pi$ is an automorphism, thus bijective, thus reordering is possible?} {\color{red} No; it is unclear what you are reordering and how.} {\color{green} Hopefully, this explanation is better.}{\color{red} Yes, that works.} 
\end{comment}
    Finally, the sum
    \begin{align*}
    &\sum_{w'_{\Gc,r},w'_{>r}}\Bigg(\pp(g^{coef}_{A,r}(U_\Gc)+W_r=w'_{\Gc,r},W_{>r}=w'_{>r})\\
    &\cdot \log_2\pp(g^{coef}_{A,r}(U_\Gc)+W_r=w'_{\Gc,r}\mid W_{>r}=w'_{>r})\Bigg)
    \end{align*}
    is equal to $-H(g^{coef}_{A,r}(U_\Gc)+W_r\mid W_{>r})$.
    \end{proof}

    In particular, this means that $d(U_\Gc,g^{coef}_{A,r}(U_\Gc))\le 2d(U_\Gc,W_r\mid W_{>r})$ by the triangle inequality. So, if $d(W_r\mid W_{>r},W_r\mid W_{>r})$ is small then $W_r\mid W_{>r}$ must be close to the uniform distribution on a subspace which is approximately preserved by all such transformations. So, our next order of business is to investigate which subspaces have this property.
    
    \subsubsection{Small orbit localization lemma and its corollary for subspace distance}\label{Sec5.2.2}
    Let $m \in \mN,\, r \in \mZ$ satisfy $0 \leq r \leq m$. Let $\Pc_{m,\leq r}=\{P \in \Pc_m \mid deg(P) \leq r\}$ be the space of polynomials of degree $\leq r$ and $\Pc_{m,r}=\Pc_{m,\leq r} /\Pc_{m,\leq r-1}$ be the quotient space of polynomials of degree $r$ ($\dim \Pc_m=2^m,\,\, \dim \Pc_{m,r}=\binom{m}{r}$). 
    %Note that $coef$ acts as a natural isomorphism between $\Pc_m$ and $\mathbb{F}_2^{2^{[m]}}$ that maps a function to a vector of its coefficients. 
    As such, there is a family of isomorphisms on $\Pc_{m,r}\rightarrow \Pc_{m,r}$ equivalent to the family of isomorphisms on $\mathbb{F}_2^{\binom{[m]}{r}}\rightarrow \mathbb{F}_2^{\binom{[m]}{r}}$ induced by linear transformations. We call this family of linear transformations $Sym(m,r)$ in the domain of $\Pc_{m,r}$ and $\overline{Sym}(m,r)=\left\{g_{A,r}^{coef} \,\,\Big|\,\, A \in GL_m(\F_2)\right\}$ in the space of $\mathbb{F}_2^{\binom{[m]}{r}}$.

    \begin{claim}
        Let $m \in \mN,\, r \in \mZ$ satisfy $0 \leq r \leq m$. Let $\Gc \subseteq \Pc_{m,r}$ be a linear space such that $\Gc=\pi(\Gc)\,\, \text{ for all } \pi \in Sym(m,r)$. Then either $\Gc=\{0\}$ or $\Gc=\Pc_{m,r}$.
    \end{claim}

    \begin{proof}
        Let $T_{i,j}$ be the transformation in $Sym(m,r)$ induced by the permutation of $x_i$ and $x_j$, and $T_{i,j}'$ be the transformation in $Sym(m,r)$ induced by the linear transformation that maps $x_i$ to $x_i+x_j$ and $x_j$ to $x_i$. Note that
        \[f \in \Gc \Rightarrow T_{i,j}f,\,\, T_{i,j}'f \in \Gc \Rightarrow T_{i,j}f+T_{i,j}'f \in \Gc.\]
        Examine $T_{i,j}f+T_{i,j}'f$ by looking at individual monomials:
        
        $(T_{i,j}+T_{i,j}')x_I=0$ if $i,j \notin I$. In fact, $x_I$ is not changed by $T_{i,j}$ or by $T_{i,j}'$.

        $(T_{i,j}+T_{i,j}')x_I=0$ if $i,j \in I$. In fact, $T_{i,j}$ transforms $x_I$ into $x_I$, as it permutes two coordinates inside it. $T_{i,j}'$ changes $x_I$ into $x_I$, as in the space $F_m$, the transformation that induces $T_{i,j}$ would transform $x_I$ into $x_I+x_{I \setminus \{j\}}$, with the second term ignored in $\Pc_{m,r}$ as it is of order $r-1$.

        $(T_{i,j}+T_{i,j}')x_I=x_I \text{ if } i \in I, j \notin I$. In fact, $T_{i,j}$ transforms $x_I$ into $x_{I \cup \{j\} \setminus \{i\}}$, and $T_{i,j}'$ transforms $x_I$ into $x_{I \cup \{j\} \setminus \{i\}}+x_I$.

        $(T_{i,j}+T_{i,j}')x_I=0 \text{ if } i \notin I, j \in I$. In fact, $T_{i,j}$ and $T'_{i,j}$ both transform $x_I$ into $x_{I \cup \{i\} \setminus \{j\}}$.

        Assume $\Gc \neq \{0\}$. This means $\exists f: f \in \Gc$, for which $\exists I: coef(f)_{x_I}=1$. As such,
        \[\left(\prod_{i \in I, j \notin I}(T_{i,j}+T_{i,j}')\right)f=x_I.\]
        This is due to every transformation either preserving or erasing a monomial, and for every monomial except $x_I$, there exists a pair $(i,j)$ such that the transformation erases this monomial. This implies that $\exists I: x_I \in \Gc$. Finally, note that $\forall J \subseteq [m] \text{ with } |J|=r, \,\,\exists \pi \in Sym(m,r): \pi(x_I)=x_J$, thus $\Gc$ is a space containing $\{x_J\mid  |J|=r\}$, which implies $\Gc=\Pc_{m,r}$.
    \end{proof}

    This claim plays a significant role in the following lemma. Define 
    \[\dist(A,B):=2\dim(A+B)-\dim(A)-\dim(B)=\dim(A)+\dim(B)-2\dim(A\cap B).\]

\begin{comment}
{\color{violet} Here is the generalized version of lemma 5.17 and its corollary. It would replace lemma 5.17 and then its original corollary would be the result of applying the new corollary to $\mathbb{F}^n=\Pc_{m,r}$ and $\mathcal{T}=Sym(m,r)$ and using the previous lemma to argue that $\Gc^\star$ must be either $\{0\}$ or $\Pc_{m,r}$.}
\end{comment}
\begin{lemma}\label{nearSymSubspace}{\bf (Small orbit localization lemma)}
Let $n\in \mN$, $\mathbb{F}$ be a finite field, $\mathcal{T}$ be a set of linear transformations on $\mathbb{F}^n$, and $\mathcal{W}$ be a probability distribution over subspaces of $\mathbb{F}^n$ such that for every $T \in \mathcal{T}$ and every subspace $\Gc_0$ of $\mathbb{F}^n$, the following equality is true: \[\mathbb{P}_{\Gc\sim\mathcal{W}}[\Gc=\Gc_0]=\mathbb{P}_{\Gc\sim\mathcal{W}}[\Gc=T\Gc_0].\]
        Then there must exist a subspace $\Gc^\star$ of $\mathbb{F}^n$ such that $T\Gc^\star=\Gc^\star$ for all $T\in\mathcal{T}$ and
        \[\mathbb{E}_{\Gc\sim\mathcal{W}}[\dist(\Gc,\Gc^\star)]\le\frac{9}{2}\mathbb{E}_{\Gc,\Gc'\sim\mathbb{\mathcal{W}}}[\dist(\Gc,\Gc')]\]
\begin{comment}
        {\color{red} In response to the question about whether $\mathbb{F}$ needs to be a finite field: The core argument still seems to be independent of that, but once there are an infinite number of subspaces of $\mathbb{F}^n$ the requirement that $\mathbb{P}_{\Gc\sim\mathcal{W}}[\Gc=\Gc_0]=\mathbb{P}_{\Gc\sim\mathcal{W}}[\Gc=T\Gc_0]$ is not necessarily meaningful. So, I would need to amend that requirement in some manner. If the subgroup of $GL_n(\mathbb{F})$ induced by $\mathcal{T}$ is finite then I could just require that the support of $\mathcal{W}$ be finite, but this feels like it should still work in some cases where that subgroup is infinite. My first instinct for dealing with that would be to require that $\mathbb{P}_{\Gc\in\mathcal{W}}[\Gc\in S]=\mathbb{P}_{\Gc\in\mathcal{W}}[T\Gc\in S]$ for arbitrary $T\in\mathcal{T}$ and set $S$, but I suspect that that could run into trouble with the Banach-Tarski paradox or comparable weirdness.}
\end{comment}
    \end{lemma}
\begin{comment}
    {\color{red} \begin{remark}
    This could be generalized to infinite fields $\mathbb{F}$ but that is beyond the scope of this paper.
    %one would likely have to generalize the requirement that
    %\[\mathbb{P}_{\Gc\sim\mathcal{W}}[\Gc=\Gc_0]=\mathbb{P}_{\Gc\sim\mathcal{W}}[\Gc=T\Gc_0].\]
    %to address the possibility that there is not a discrete set of possible values one might draw from $\mathcal{W}$.
    \end{remark}}
\end{comment}
    \begin{proof}
        For a probability distribution $\mathcal{W}$ over subspaces of $\mathbb{F}^n$, let 
        \begin{align*}
            \Delta \mathcal{W}:=\mathbb{E}_{\Gc,\Gc'\sim\mathcal{W}}[\dist(\Gc,\Gc')],\,\, \dim(\mathcal{W}):=\mathbb{E}_{\Gc\sim\mathcal{W}}[\dim(\Gc)].
        \end{align*}
        Also, given two probability distributions $\mathcal{W}$ and $\mathcal{W}'$ over subspaces of $\mathbb{F}^n$, let the distance between them be the minimum over all couplings\footnote{Given random variables $A$ and $B$ a coupling between them is a probability distribution over $(A,B)$ such that $A$ and $B$ still have their original probability distributions.} of them of the expected distance between their subspaces. In other words, we define
        \[\dist(\mathcal{W},\mathcal{W}')=\min_{p(\Gc,\Gc'): p(\Gc)= \mathcal{W}, p(\Gc')=\mathcal{W}'}\mathbb{E}[\dist(\Gc,\Gc')]\]
        $\dist(\mathcal{W},\mathcal{W}')=0$ iff $\mathcal{W}=\mathcal{W}'$ because that is the only circumstance under which we can always have $\Gc=\Gc'$, and $\dist(\mathcal{W},\mathcal{W}')=\dist(\mathcal{W}',\mathcal{W})$ for all $\mathcal{W}$ and $\mathcal{W}'$. Showing that this definition of distance satisfies the triangle inequality is a little more complicated. In order to do so, first let $\mathcal{W}$, $\mathcal{W}'$, and $\mathcal{W}''$ be probability distributions over subspaces of $\mathbb{F}^n$. Now, consider drawing $\Gc'\sim \mathcal{W}'$ and then drawing $\Gc$ and $\Gc''$ from their probability distributions conditioned on that value of $\Gc'$ under the couplings minimizing $\mathbb{E}[\dist(\Gc,\Gc')]$ and $\mathbb{E}[\dist(\Gc',\Gc'')]$ (these couplings are not necessarily unique, but we can pick one arbitrarily). Under that probability distribution of $(\Gc,\Gc',\Gc'')$ we have that:
        \begin{align*}
        \dist(\mathcal{W},\mathcal{W}'')&\le\mathbb{E}[\dist(\Gc,\Gc'')]\\
        &\le \mathbb{E}[\dist(\Gc,\Gc')+\dist(\Gc',\Gc'')]\\
        &=\mathbb{E}[\dist(\Gc,\Gc')]+\mathbb{E}[\dist(\Gc',\Gc'')]\\
        &=\dist(\mathcal{W},\mathcal{W}')+\dist(\mathcal{W}',\mathcal{W}'')
        \end{align*}
        So, this definition of the distance between two probability distributions over subspaces has all of the necessary properties. 
        %{\color{red} This definition of the distance between two probability distributions over subspaces is new, and my explanation might need work.}

        We will show that for any $\mathcal{W}$ there exists a $\mathcal{W}'$ close to $\mathcal{W}$ with $\Delta \mathcal{W}'$ significantly smaller than $\Delta \mathcal{W}$ and then argue that repeated substitution must eventually yield a probability distribution that returns a subspace that is preserved by all $T\in \mathcal{T}$ with high probability. As such, a random $\Gc\sim \mathcal{W}$ must be close to this subspace in expectation.

        In order to do that, first let $d_i=\mathbb{E}_{\Gc_1,...,\Gc_{i+1}\sim \mathcal{W}}[\dim(\Gc_1+...+\Gc_{i+1})-\dim(\Gc_1+...+\Gc_i)]$ for all $i$, and note that $\dim(\mathcal{W})=d_0$ and $\Delta \mathcal{W}=2 d_1$. For our first candidate for $\mathcal{W}'$, let $\mathcal{W}^\star$ be the probability distribution of $\Gc_1+\Gc_2$ when $\Gc_1,\Gc_2\sim\mathcal{W}$. Note that
        \begin{align*}
        \Delta \mathcal{W}^\star&=\mathbb{E}_{\Gc_1,...,\Gc_4\sim\mathcal{W}} \left[ \dim(\Gc_1+\Gc_2)+\dim(\Gc_3+\Gc_4)-2\dim((\Gc_1+\Gc_2)\cap(\Gc_3+\Gc_4))\right] \\
        &=\mathbb{E}_{\Gc_1,...,\Gc_4\sim\mathcal{W}} \left[2\dim(\Gc_1+\Gc_2+\Gc_3+\Gc_4)-\dim(\Gc_1+\Gc_2)-\dim(\Gc_3+\Gc_4)\right]\\
        &=2d_3+2d_2.
        \end{align*}
        Also,
        \begin{align*}
        \dist(\mathcal{W}^\star,\mathcal{W})&\le \mathbb{E}_{\Gc_1,\Gc_2\sim\mathcal{W}}[\dist(\Gc_1+\Gc_2,\Gc_1)]\\
        &=\mathbb{E}_{\Gc_1,\Gc_2\sim\mathcal{W}}[\dim(\Gc_1+\Gc_2)-\dim(\Gc_1)]\\
        &=d_1
        \end{align*}
        As such, if $d_2+d_3$ is significantly less than $d_1$ this would be suitable for $\mathcal{W}'$. In order to cover the case where it is not, let $\mathcal{W}^{(i)}$ be the probability distribution of $\Gc_0\cap(\sum_{j=1}^i \Gc_j)$ when $\Gc_0,...,\Gc_i\sim \mathcal{W}$. For each $i$,
        \[\scalebox{0.81}{$\Delta\mathcal{W}^{(i)}:=\mathbb{E}_{\Gc_0,...\Gc_i,\Gc'_0,..,\Gc'_i\sim\mathcal{W}}\left[ 2\dim\left(\Gc_0\cap\left(\sum_{j=1}^i \Gc_j\right)\right)-2\dim\left(\Gc_0\cap \Gc'_0 \cap\left(\sum_{j=1}^i \Gc_j\right)\cap \left(\sum_{j=1}^i \Gc'_j\right)\right)\right]$}.\]
        In order to bound that expression, first observe that the following is true: 
        \[\scalebox{0.93}{$\mathbb{E}\left[\dim\left(\Gc_0\cap\left(\sum_{j=1}^i \Gc_j\right)\right) \right]=\mathbb{E}\left[\dim(\Gc_0)+\dim\left(\sum_{j=1}^i \Gc_j\right)-\dim\left(\sum_{j=0}^i \Gc_j\right) \right]=d_0-d_i$}\]
        and
        \begingroup\makeatletter\def\f@size{9.}\check@mathfonts
        \def\maketag@@@#1{\hbox{\m@th\large\normalfont#1}}%
        \begin{align*}
        &\E\left[ \dim\left(\Gc_0\cap \Gc'_0 \cap\left(\sum_{j=1}^i \Gc_j\right)\cap \left(\sum_{j=1}^i \Gc'_j\right)\right)\right]\\
        &= \E\left[2 \dim\left(\Gc_0\cap \Gc'_0 \cap\left(\sum_{j=1}^i \Gc_j\right)\right)-\dim \left(\Gc_0\cap \Gc'_0 \cap\left(\sum_{j=1}^i \Gc_j\right)+\Gc_0\cap \Gc'_0\cap \left(\sum_{j=1}^i \Gc'_j\right)\right)\right] \\
        &\ge \E\left[2 \dim\left(\Gc_0\cap \Gc'_0 \cap\left(\sum_{j=1}^i \Gc_j\right)\right)-\dim \left(\Gc_0\cap \Gc'_0\right)\right]\\
        &= \E\left[\dim\left(\Gc_0\cap \Gc'_0\right)+2 \dim\left(\Gc'_0 \cap\left(\sum_{j=1}^i \Gc_j\right)\right)-2 \dim\left(\Gc_0\cap \Gc'_0+\Gc'_0 \cap\left(\sum_{j=1}^i \Gc_j\right)\right)\right] \\
        &\ge \E\left[\dim\left(\Gc_0\cap \Gc'_0\right)+2 \dim\left(\Gc'_0 \cap\left(\sum_{j=1}^i \Gc_j\right)\right)-2 \dim\left(\Gc'_0 \cap\left(\sum_{j=0}^i \Gc_j\right)\right)\right]\\
        &= (d_0-d_1)+2(d_0-d_i)-2(d_0-d_{i+1})=d_0-d_1-2(d_i-d_{i+1}).
        \end{align*}\endgroup
        Putting these together, we get that 
        \[\Delta\mathcal{W}^{(i)}\le 2(d_0-d_i)-2[d_0-d_1-2(d_i-d_{i+1})]=2d_1+2d_i-4d_{i+1}.\]
        Also,
        \begin{align*}
        \dist(\mathcal{W}^{(i)},\mathcal{W})
        &\le \mathbb{E}_{\Gc_0,...,\Gc_i}[\dist(\Gc_0,\Gc_0\cap(\sum_{j=1}^i\Gc_j)]\\
        &=\mathbb{E}_{\Gc_0,...,\Gc_i}[\dim(\Gc_0)-\dim(\Gc_0\cap(\sum_{j=1}^i\Gc_j))]\\
        &=d_i
        \end{align*}
        
        Finally, if $d_2\ge (5/9)d_1$ then $\Delta\mathcal{W}^{(1)}\le 4d_1-4d_2\le (16/9)d_1=(8/9)\Delta\mathcal{W}$. If $d_2< (5/9)d_1$ and $d_3\ge (1/3)d_1$ then $\Delta\mathcal{W}^{(2)}\le 2d_1+2d_2-4d_3\le (16/9)d_1=(8/9)\Delta\mathcal{W}$. Otherwise, $d_2< (5/9)d_1$ and $d_3<(1/3)d_1$, in which case $\Delta\mathcal{W}^\star=2d_2+2d_3<(8/9)\Delta\mathcal{W}$. In all three cases, this gives us a $\mathcal{W}'$ such that $\mathcal{W}'$ is also preserved by all $T\in \mathcal{T}$, $\Delta\mathcal{W}'\le (8/9)\Delta\mathcal{W}$, and $\dist(\mathcal{W},\mathcal{W}')\le\max(d_1,d_2)=d_1=\Delta\mathcal{W}/2$.
        
        That means that there exists an infinite sequence of probability distributions over subspaces $\mathcal{W}_0,\mathcal{W}_1,...$ that satisfies the following properties:
        \begin{enumerate}
        \item $\mathcal{W}_0=\mathcal{W}$.
        \item $\mathcal{W}_i$ is preserved by all $T\in \mathcal{T}$ for all $i\ge 0$.
        \item $\Delta\mathcal{W}_{i+1}\le (8/9)\Delta\mathcal{W}_i$ for all $i$.
        \item $\dist(\mathcal{W}_{i+1},\mathcal{W}_i)\le\Delta\mathcal{W}_i/2$ for all $i$.
        \end{enumerate}
        
        Next, observe that $\Delta\mathcal{W}_i=\mathbb{E}_{\Gc,\Gc'\sim \mathcal{W}_i}[\dist(\Gc,\Gc')]\ge \mathbb{P}_{\Gc,\Gc'\sim \mathcal{W}_i}[\Gc\ne \Gc']$ for all $i$. So, for all sufficiently large $i$, $\mathcal{W}_i$ returns one subspace with high probability. Furthermore, $\dist(\mathcal{W}_i,\mathcal{W}_{i+1})$ is at least equal to the difference in their probabilities of returning that subspace, so it must be the same subspace for all sufficiently large $i$. Call it $\Gc^\star$ and observe that it must be preserved by all $T\in \mathcal{T}$. Finally,
        \begin{align*}
        \mathbb{E}_{\Gc\sim \mathcal{W}}[\dist(\Gc,\Gc^\star)]&=\lim_{i\to\infty}\dist(\mathcal{W},\mathcal{W}_i)\\
        &\le\sum_{j=0}^\infty (8/9)^j\Delta\mathcal{W}/2\\
        &=(9/2)\Delta\mathcal{W}   
        \end{align*}
    \end{proof}

        \begin{corollary}\label{cor:general_dist_bound}
        Let $n\in \mN$, $\mathbb{F}$ be a finite field, and  $\mathcal{T}$ be a group of linear transformations on $\mathbb{F}^n$. For every subspace $\Gc\subseteq \mathbb{F}^n$, there exists a subspace $\Gc^\star\subseteq\mathbb{F}^n$ such that $T\Gc^\star=\Gc^\star$ for all $T\in\mathcal{T}$ and $\dist(\Gc,\Gc^\star)\le (9/2)\max_{T\in\mathcal{T}} \dist(\Gc,T\Gc)$.
    \end{corollary}

    \begin{corollary}\label{cor:dist_bound}
        Let $m \in \mN,\, r \in \mZ$ satisfy $0 \leq r \leq m$. For every subspace $\Gc\subseteq \Pc_{m,r}$, there exists $\pi \in \overline{Sym}(m,r): \dist(\Gc, \pi(\Gc)) \geq \frac{2}{9}\min(\dim(\Gc),\binom{m}{r}-\dim(\Gc)).$
    \end{corollary}
        \subsubsection{Linking permutation invariance to bounds on the entropy of sums}
        \begin{theorem}[Recurrent layer entropy inequality]
            Let $m \in \mN,\, r \in \mZ$ satisfy $0 \leq r \leq m$. The following layer entropy inequality holds:
            \begin{equation}\label{ineq:recineq}
            \begin{split}
                &\frac{1}{140}\min\left(H(W_r\mid W_{>r}),\binom{m}{r}-H(W_r\mid W_{>r})\right)+H(W_r\mid W_{>r}) \\
                &\leq H(W_r+W'_r\mid W_{>r},W'_{>r})
            \end{split}
            \end{equation}
        \end{theorem}
        \begin{proof}
        Define $\upsilon:=d(W_r\mid W_{>r},W_r\mid W_{>r})$.
        Corollary \ref{CEFR} and Lemma \ref{perm} imply the existence of a subspace $\Gc\subseteq \Pc_{m,r}$ such that for all $\pi\in Sym(m,r)$
        \[d(\pi(U_\Gc),W_r\mid W_{>r})=d(U_\Gc,W_r\mid W_{>r}) \leq 7\upsilon.\]
        Using the triangle inequality for the Ruzsa distance,
        \[d(\pi(U_\Gc),U_\Gc) \leq d(\pi(U_\Gc),W_r\mid (W_{>r}=w_{>r}))+d(U_\Gc,W_r\mid (W_{>r}=w_{>r})).\]
        for all $w_{>r}$. Taking the expectation with respect to $w_{>r}$, we obtain
        \[d(\pi(U_\Gc),U_\Gc) \leq d(\pi(U_\Gc),W_r\mid W_{>r})+d(U_\Gc,W_r\mid W_{>r}) \leq 14\upsilon.\]
        Note that
        \[d(\pi(U_\Gc),U_\Gc)=H(U_{\Gc+\pi(\Gc)})-H(U_\Gc)=\dim(\Gc+\pi(\Gc))-\dim(\Gc)=\frac{1}{2}\dist(\Gc,\pi(\Gc)).\]
        From the last subsubsection, we know that there exists $\pi \in \overline{Sym}$ such that
        \[\dist(\Gc,\pi(\Gc))\geq \frac{2}{9} \min \left(\dim(\Gc),\binom{m}{r}-\dim(\Gc)\right).\]
        This gives us the following bound on $\dim(\Gc)$:
        \[\frac{1}{9} \min \left(\dim(\Gc),\binom{m}{r}-\dim(\Gc)\right) \leq d(\pi(U_\Gc),U_\Gc) \leq 14\upsilon.\]
        Using $d(X,Y)\geq \frac{1}{2}|H(X)-H(Y)|$, we get
        \[\frac{1}{2}|\dim(\Gc)-H(W_r\mid W_{>r})|=\frac{1}{2}|H(U_\Gc)-H(W_r\mid W_{>r})|\leq d(U_\Gc,W_r\mid W_{>r}) \leq 7\upsilon.\]
        To conclude, note that
        \[|\min(a,c-a)-\min(b,c-b)|=\frac{1}{2}||c-2a|-|c-2b|| \leq |a-b|.\]
        Setting $a=\dim(\Gc),\,\, b=H(W_r\mid W_{>r}),\,\, c=\binom{m}{r}$, we show
        \[\scalebox{0.95}{$\left|\min\left(\dim(\Gc),\binom{m}{r}-\dim(\Gc)\right)-\min\left(H(W_r\mid W_{>r}),\binom{m}{r}-H(W_r\mid W_{>r})\right)\right| \leq 14 \upsilon.$}\]
        Finally, this implies
        \begin{align*}
        &\min\left(H(W_r\mid W_{>r}),\binom{m}{r}-H(W_r\mid W_{>r})\right) \\
        &\leq \min\left(\dim(\Gc),\binom{m}{r}-\dim(\Gc)\right)+14\upsilon \leq 140\upsilon.
        \end{align*}
        \end{proof}
    \subsubsection{From the recurrence bound to a polarization bound}\label{Sec6}
        In this section, we use $W^{(m)}_r$ instead of $W_r$, where the superscript $(m)$ indicates the dependence on the parameter $m$. Recall that $Z \sim Ber(\delta)^{\F_2^m}$, where $\delta \in (0,\frac{1}{2})$ represents the error probability. Define 
        \[H(\delta):=-\delta \log_2(\delta) - (1-\delta) \log_2(1-\delta).\]
        %{\color{olive} Suggestion: move this theorem to the previous section and write this is the proof of part (2) Theorem 3.1.} {\color{teal} I've arranged it a bit differently so that the entire Proof is in the same section. If this doesn't work, I will move this theorem.}
        \begin{theorem}\label{thm: recineq}
            Let $m,\, n \in \mN,\, r \in \mZ$ satisfy $0 \leq r \leq m,\, n=2^m$. $m$ will be considered a varying parameter here. Denote $f_{m,r}:=H(W^{(m)}_r\mid W^{(m)}_{>r})$ and $a_{m, r}:=H(W^{(m)}_{\leq r}\mid W^{(m)}_{>r})$. The following layer polarization inequality holds:
            \begin{equation}\label{eq_main_thm_part_1}
               a_{m+1,r} \leq a_{m, r}+a_{m, r-1}-\frac{1}{140}\min(f_{m,r},\binom{m}{r}-f_{m,r})
            \end{equation}
        \end{theorem}
        \begin{proof}
        Consider the following equality: $g(x_1,x_2 \ldots x_m)=g(x_1,x_2 \ldots x_{m-1},0)+x_m (g(x_1,x_2 \ldots x_{m-1},0)+g(x_1,x_2 \ldots x_{m-1},1))$ for any function $g$ such that \\$eval(g(x_1,...,x_m))\in RM(m,r)$. Note that
        \begin{align*}
            &eval(g(x_1,x_2 \ldots x_{m-1},0)) \in RM(m-1,r);\\
            &eval(g(x_1,x_2 \ldots x_{m-1},0)+g(x_1,x_2 \ldots x_{m-1},1)) \in RM(m-1,r-1).
        \end{align*}
        As such, there is a connection between $RM(m,r)$ and $(RM(m-1,r),RM(m-1,r-1))$, which translates into a generating matrix recurrence as follows:
    
        \[G_{m,r}=\left(\begin{matrix} G_{m-1,r} & 0 \\ G_{m-1,r} & G_{m-1, r-1} \end{matrix}\right), \,\,G_{m,m}=\left(\begin{matrix}
            1 & 0 \\ 1 & 1
        \end{matrix}\right)^{\otimes m}, \,\,G_{m,0}=\left(\begin{matrix}
            1 \\ 1
        \end{matrix}\right)^{\otimes m}.\]
    
        Additionally, define $\overline{G_{m,r}}$ as a matrix with columns being evaluations of monomials of degree at least $r+1$. Let $W^{(m)}$ and $W'^{(m)}$ be two independent instances of $W^{(m)}$. Analogously to $G_{m,r}$, $\overline{G_{m,r}}$ adheres to the following recurrent relation:
    
        \[\overline{G_{m,r}}=\left(\begin{matrix} \overline{G_{m-1,r}} & 0 \\ \overline{G_{m-1,r}} & \overline{G_{m-1,r-1}} \end{matrix}\right).\]
        Considering $H(W^{(m+1)}_{\leq r}\mid W^{(m+1)}_{>r})$, note that $W^{(m+1)}_{\leq r}$ is a permuted version of \\$(G_{m+1,m+1}Z^{(m+1)})_{\leq r}=\overline{G_{m+1,m-r}}^T Z^{(m+1)}$ and $W^{(m+1)}_{> r}$ is a permuted version of \\$(G_{m+1,m+1}Z^{(m+1)})_{> r}=G_{m+1,m-r}^T Z^{(m+1)}$. As such, the recurrent relations for $\overline{G_{m,r}}, G_{m,r}$ imply
        \begin{align*}
            &a_{m+1,r}=H(W^{(m+1)}_{\leq r}\mid W^{(m+1)}_{>r})=H(\overline{G_{m+1,m-r}}^T Z^{(m+1)}\mid G_{m+1,m-r}^T Z^{(m+1)})\\
            &=H\left(\left(\begin{matrix} \overline{G_{m,m-r}}^T & \overline{G_{m,m-r}}^T \\ 0  & \overline{G_{m,m-r-1}}^T \end{matrix}\right)\left(\begin{matrix} Z^{(m)} \\ Z'^{(m)} \end{matrix}\right)\,\Bigg|\,\left(\begin{matrix} G_{m,m-r}^T & G_{m,m-r}^T \\ 0 & G_{m, m-r-1}^T \end{matrix}\right)\left(\begin{matrix} Z^{(m)} \\ Z'^{(m)} \end{matrix}\right)\right)\\
            &=H\left(W_{\leq r-1}^{(m)}+W_{\leq r-1}'^{(m)},W_{\leq r}'^{(m)}\,\big|\,W_{>r-1}^{(m)}+W_{>r-1}'^{(m)},W_{>r}'^{(m)}\right)\\
            &=H\left(W_{\leq r-1}^{(m)},W_{\leq r}'^{(m)}\,\big|\,W_{>r}^{(m)},W_{>r}'^{(m)},W_{r}^{(m)}+W_{r}'^{(m)}\right)\\
            &=H\left(W_{\leq r}^{(m)},W_{\leq r}'^{(m)}\,\big|\,W_{>r}^{(m)},W_{>r}'^{(m)}\right)-H\left(W_{r}^{(m)}+W_{r}'^{(m)}\,\big|\,W_{>r}^{(m)},W_{>r}'^{(m)}\right)\\
            &\leq 2a_{m, r}-f_{m,r}-\frac{1}{140}\min(f_{m,r},\binom{m}{r}-f_{m,r})\\
            &=a_{m, r}+a_{m, r-1}-\frac{1}{140}\min(f_{m,r},\binom{m}{r}-f_{m,r}).
        \end{align*}
        \begin{itemize}
            \item The third-to-last line follows from $H(A\mid B,C)=H(A,B\mid C)-H(B\mid C)$ for \[A=(W_{\leq r-1}^{(m)},\,\,W_{\leq r}'^{(m)}),\,\, B=W_{r}^{(m)}+W_{r}'^{(m)},\,\, C=(W_{>r}^{(m)},W_{>r}'^{(m)}).\]
            \item The inequality follows from the left term being equal to $2a_{m, r}$ and the right term bounded in the previous section.
            \item The last equation follows from $f_{m,r}=a_{m, r} - a_{m, r-1}$.
        \end{itemize}
        \end{proof}
        This concludes the proof of Theorem 3.1 (1).
    \subsection{Proof of Theorem 3.1 (2): From the polarization bound to an entropy bound}

%Suggestion: Add a paragraph describing this section. In this section we use inequality \ref{eq_main_thm_part_1}  together with some other inequalities satisfied by the sequence $a_{m,r}$ and obtain the upper bound of Theorem 3.1 (2).

In this section, we focus on the double indexed sequence of numbers $(a_{m,r})_{m\in\mathbb{N},r\in\{0,\ldots,m\}}$. As the numbers $a_{m,r}$ represent entropies, they satisfy the inequalities 
\begin{equation}\label{eqn: amr easy 1}
0\leq a_{m,r}\leq \binom{m}{\leq r}
\end{equation}
and the equalities
\begin{equation}\label{eqn: amr easy 2}
a_{m,m}=2^m H(\delta).
\end{equation}
In addition,  we consider a corollary of a partial order of entropies theorem from \cite{Ye19}, which establishes the monotonicity of normalized layer entropies.
\begin{theorem}[Layer monotonicity \cite{Ye19}]\label{corr2}
     Let $m,\, r \in \mathbb{Z}$ satisfy $0 \leq r < m$. Let $f_{m,r}^{avg}:=\frac{H(W_r \mid W_{>r})}{\binom{m}{r}}$ for $r \in \{0\} \cup [m]$. Then $f^{avg}_{m+1,r} \leq f^{avg}_{m,r} \leq f^{avg}_{m,r+1}$.
\end{theorem}
Finally, the numbers $a_{m,r}$ satisfy the inequality of Theorem~\ref{thm: recineq}. In this section we show that the above inequalities are sufficient to deduce the upper bound of Theorem~\ref{thm: final_res_p2}.
    
Our first goal is to write the inequality of Theorem \ref{thm: recineq} as a recurrence relation on entropies $a_{m,r}$, satisfied under additional constraints on parameters $m$ and $r$. Then, we define a stochastic process on $r$ that, given the linear recurrent inequality, defines a submartingale $\frac{a_{m-t,r}}{2^{m-t}}$ with a discrete time parameter $t$. Finally, we show that the probability distribution of $r$ is concentrated around values of $r$ that correspond to small values of $a_{m,r}$, which implies a small upper bound on $a_{m,r}$.
    
    We define the following parameter $r$ for variation: 
    \[r(m,\epsilon):=\max \{r\mid  f^{avg}_{m,r}<1-\epsilon\}.\]
   
    By Theorem \ref{corr2},
    \[\forall r \leq r(m,\epsilon): f^{avg}_{m,r}\leq 1-\epsilon.\]
    This implies that
    \begin{align*}
        &\forall r<r(m,\epsilon): \frac{1}{140}min\left(f_{m,r},\binom{m}{r}-f_{m,r}\right) \geq min\left(\frac{f_{m,r}}{140},\frac{\epsilon}{140}\binom{m}{r}\right) \geq \frac{\epsilon}{140}f_{m,r}, \\
        &f_{m+1,\leq r} \leq a_{m, r}+a_{m, r-1}-\frac{\epsilon}{140}f_{m,r}=\left(1-\frac{\epsilon}{140}\right)a_{m, r}+\left(1+\frac{\epsilon}{140}\right)a_{m, r-1}.
    \end{align*}
    Note that the behavior of $r(m,\epsilon)$ is difficult to examine. The next claim allows us to find a lower bound for $r(m,\epsilon)$ with a clear asymptotical behavior.

    \begin{claim}
        Let $m \in \mN,\epsilon \in (0,1)$.
        \[\frac{\binom{m}{\leq r(m,\epsilon)}}{2^m} \geq 1-\frac{H(\delta)}{1-\epsilon}.\]
    \end{claim}
    \begin{proof}
        $\forall r>r(m,\epsilon): f_{m,r}^{avg}>1-\epsilon$. This implies $\frac{H(W^{(m)}_{\geq r})}{\binom{m}{\geq r}}=\frac{\sum_{i: r \leq i \leq m}f_{m,i}^{avg}\binom{m}{i}}{\sum_{i: r \leq i \leq m}\binom{m}{i}}>1-\epsilon$, and as such,
        \[\forall r>r(m,\epsilon): \,\, 2^m H(\delta) \geq a_{m,m}-a_{m,r-1} = H(W^{(m)}_{\geq r}) \geq (1-\epsilon)\binom{m}{\geq r}.\]
        Therefore,
        \[ \frac{H(\delta)}{1-\epsilon} \geq \frac{\binom{m}{\geq r}}{2^m}=1-\frac{\binom{m}{<r}}{2^m}.\]
        Using this inequality for $r=r(m,\epsilon)+1$, we get
        \[\frac{\binom{m}{\leq r(m,\epsilon)}}{2^m} \geq 1-\frac{H(\delta)}{1-\epsilon}.\]
    \end{proof}

    \begin{corollary}\label{r_star}
        Let $m \in \mN$ be a varying parameter. Let $r^*(m,\epsilon)=\min \{r \in \{0\} \cup [m]\mid \frac{\binom{m}{\leq r}}{2^m} \geq 1-\frac{H(\delta)}{1-\epsilon}\}$. The following relations are true:
        \begin{enumerate}
            \item $r^*(m,\epsilon) \leq r(m, \epsilon)$,
            \item $\lim_{m \rightarrow \infty}\frac{\binom{m}{\leq r^*(m,\epsilon)}}{2^m} = 1-\frac{H(\delta)}{1-\epsilon},$
            \item $r^*(m,\epsilon)=\frac{m}{2}+C(\epsilon)\sqrt{m}+o_m(\sqrt{m}).$
        \end{enumerate}
        Here, $C(\epsilon)$ depends on both $\epsilon, \delta$.
    \end{corollary}
    \begin{proof}
        The first property follows from the claim. To prove the second and third properties, note $\frac{\binom{m}{\leq r^*(m,\epsilon)}}{2^m} \geq 1-\frac{H(\delta)}{1-\epsilon} \geq \frac{\binom{m}{\leq r^*(m,\epsilon)-1}}{2^m}$ . Now, note that $0 \leq \frac{\binom{m}{\leq r^*(m,\epsilon)}}{2^m}-\frac{\binom{m}{\leq r^*(m,\epsilon)-1}}{2^m} \leq \frac{\binom{m}{m/2}}{2^m}=O_m(\frac{1}{\sqrt{m}})$, thus $\frac{\binom{m}{\leq r^*(m,\epsilon)}}{2^m}$ has a limit of $1-\frac{H(\delta)}{1-\epsilon}$ as it is within the $O_m(\frac{1}{\sqrt{m}})$ neighborhood of $1-\frac{H(\delta)}{1-\epsilon}$. The third property follows from the fact that $\frac{\binom{m}{\leq m/2+C_1\sqrt{m}}}{2^m}$ and $\frac{\binom{m}{\leq m/2+C_2\sqrt{m}}}{2^m}$ have different limits for $C_1 \neq C_2$.
    \end{proof}

    To proceed with the analysis, note that
    \[a_{m+1,r}\leq\left(1-\frac{\epsilon}{140}\right)a_{m, r}+\left(1+\frac{\epsilon}{140}\right)a_{m, r-1}\]
    as long as $r \leq r^*(m,\epsilon)$ and
    \[a_{m+1,r}\leq a_{m, r}+a_{m, r-1}\]
    for any $r \in [m]$, including $r > r^*(m,\epsilon)$.

    Now, we wish to keep track of coefficients of $a_{m-k,r}$ for different $r$ after applying the inequalities above to $a_{m,r}$ $k$ times. It helps to reformulate the inequalities as follows:
    \begin{align*}
        \frac{a_{m+1,r}}{2^{m+1}}\leq\left(\frac{1}{2}-\frac{\epsilon}{280}\right)\frac{a_{m, r}}{2^m}+\left(\frac{1}{2}+\frac{\epsilon}{280}\right)\frac{a_{m, r-1}}{2^m} \text{ if } r \leq r^*(\epsilon, m), \\
        \frac{a_{m+1,r}}{2^{m+1}}\leq\left(\frac{1}{2}\right)\frac{a_{m, r}}{2^m}+\left(\frac{1}{2}\right)\frac{a_{m, r-1}}{2^m} \text{ if } r > r^*(\epsilon, m).
    \end{align*}
    Note that in both inequalities, the coefficients add up to 1. So, we are able to view the coefficients as probabilities of a certain stochastic process. This is formalized by defining a discrete-time submartingale. At time 0, one can define the initial value as $\frac{a_{m, r}}{2^m}$. At time 1, one can define the value as $\frac{a_{m-1, r-\xi}}{2^{m-1}}$, where $\xi \sim \mathrm{Ber}\left(\frac{1}{2}\right)$ if $r > r^*(m-1, \epsilon)$ and $\xi \sim \mathrm{Ber}\left(\frac{1+\epsilon/140}{2}\right)$ if $r \leq r^*(m-1, \epsilon)$. This submartingale is further defined in the Definition below.

    \begin{definition}\label{stoc_pr}
        Let $m \in \mN,\,\delta\in(0,\frac{1}{2})$ be fixed parameters and $\epsilon, \omega \in (0,1)$ be varying parameters.
\begin{comment}
        Let $r_{\mathrm{cap}}: \mathbb{N}\times (0,1)\rightarrow \mathbb{R}$ and $C:(0,1)\rightarrow\mathbb{R}$ be functions such that
        \begin{align*}
        &\lim_{\epsilon \rightarrow 0+} \lim_{m \rightarrow \infty} \frac{\binom{m}{\leq r_{\mathrm{cap}}(m,\epsilon)}}{2^m}=1-H(\delta),\\
        &\text{ and }\\
        &r_{\mathrm{cap}}(m,\epsilon)=\frac{m}{2}+C(\epsilon)\sqrt{m}+o_m(\sqrt{m}).
        \end{align*} for any fixed $\epsilon\in (0,1)$.
\end{comment}
        Let $\Delta:\mN \times(0,1)\rightarrow (0,1)$ be a function of $m \in \mN$ and $\epsilon \in (0,1)$ such that $a_{m+1,r}\leq\left(1-\Delta_m(\epsilon)\right)a_{m, r}+\left(1+\Delta_m(\epsilon)\right)a_{m, r-1}$ for all $m \in \mN,\,\epsilon\in(0,1)$ and $r \leq r^*(m,\epsilon)$. Finally, let $C(\epsilon)$ be the constant defined in Corollary \ref{r_star}.
        The following stochastic process representing the entropy of the code normalized by the length of the codeword is defined:
        \[\kappa^{(m)}_{k}:=\frac{a_{m-k,\left\lfloor m/2+(C(\epsilon)-\omega)\sqrt{m}\right\rfloor+(\zeta^{(m)}_{k}-k)/2}}{2^{m-k}},\]
        where $\zeta_{k}^{(m)}=\sum_{i=1}^{k} \xi_i^{(m)}, \,\, \zeta_0^{(m)}=0,\,\, \Big(\xi_i^{(m)}\,\Big|\,\zeta_{i-1}^{(m)}=t\Big) \sim 2\mathrm{Ber}\Big(\frac{1}{2}\Big)-1 \text{ if } t\geq \frac{\omega \sqrt{m}}{2}$ and $2\mathrm{Ber}\Big(\frac{1-\Delta_m(\epsilon)}{2}\Big)-1 \text{ if } t<\frac{\omega \sqrt{m}}{2}$.% {\color{red} Wait, this is ultimately defining $\kappa$, $\zeta$, and $\xi$, right? Why is $r_{cap}$ necessary to define those? You probably should not have an unspecified constant at the end of the definition.} {\color{teal} Rewrote.}
    \end{definition}
%{\color{olive} Suggestion: Add motivation for defining the stochastic process below.}
    Note that we have proven that $\Delta_m(\epsilon)=\frac{\epsilon}{140}$ satisfy the restrictions of the definition \ref{stoc_pr}.
    %In the following, we define a stochastic process representing the entropy of the code normalized by the length of the codeword:
    %\[\kappa^{(m)}_{k}:=\frac{a_{m-k,\left\lfloor m/2+(C(\epsilon)-\omega)\sqrt{m}\right\rfloor+(\zeta^{(m)}_{k}-k)/2}}{2^{m-k}},\]
    %where $\zeta_{k}^{(m)}=\sum_{i=1}^{k} \xi_i^{(m)}, \,\, \zeta_0^{(m)}=0,\,\, \Big(\xi_i^{(m)}\,\Big|\,\zeta_{i-1}^{(m)}=t\Big) \sim \mathrm{Rad}\Big(\frac{1}{2}\Big) \text{ if } t\geq \frac{\omega \sqrt{m}}{2}$ and $\mathrm{Rad}\Big(\frac{1}{2}-\frac{\epsilon}{280}\Big) \text{ if } t<\frac{\omega \sqrt{m}}{2}$. Here, $\omega$ is an arbitrarily small positive constant.
    
    \begin{property}
        Let $m \in \mN$ be a varying parameter. $\kappa_i^{(m)} \leq \E\Big(\kappa^{(m)}_{i+1}\,\Big|\,\zeta^{(m)}_i\Big)\,\, \text{ when } i<cm$ for a small enough constant $c$ and a large enough $m$.
    \end{property}
    \begin{proof}
        
        Note that if $\zeta^{(m)}_i=t$ then $\kappa_i^{(m)}=\frac{a_{m-i,\left\lfloor\frac{m}{2}+(C(\epsilon)-\omega)\sqrt{m}\right\rfloor+\frac{t-i}{2}}}{2^{m-i}}$. If $c$ is small enough and $m$ is large enough,
        \[\left\lfloor\frac{m}{2}+(C(\epsilon)-\omega)\sqrt{m}\right\rfloor+\frac{\frac{\omega \sqrt{m}}{2}-i}{2} \leq \frac{m-i}{2}+\left(C(\epsilon)-\frac{3\omega}{4}\right)\sqrt{m} < r_{\mathrm{cap}}(m-i,\epsilon).\]
        Thus, if $t<\frac{\omega \sqrt{m}}{2}$, the following is true:
        
        \begin{align*}
        &(\kappa_i^{(m)}\mid \zeta^{(m)}_i=t)\\
         & \leq \left(\frac{1-\Delta_m}{2}\right)\Big(\kappa_{i+1}^{(m)}\,\Big|\,\zeta^{(m)}_{i+1}=t+1\Big)+\left(\frac{1+\Delta_m}{2}\right)\Big(\kappa_{i+1}^{(m)}\,\Big|\,\zeta^{(m)}_{i+1}=t-1\Big)\\
        &= \pp\Big(\xi^{(m)}_{i+1}=1\,\Big|\,\zeta^{(m)}_i=t\Big)\Big(\kappa_{i+1}^{(m)}\,\Big|\,\zeta^{(m)}_{i+1}=t+1\Big)\\
        &+\pp\Big(\xi^{(m)}_{i+1}=-1\,\Big|\,\zeta^{(m)}_i=t\Big)\Big(\kappa_{i+1}^{(m)}\,\Big|\,\zeta^{(m)}_{i+1}=t-1\Big)\\
        &=\pp\Big(\zeta^{(m)}_{i+1}=t+1\,\Big|\,\zeta^{(m)}_i=t\Big)\Big(\kappa_{i+1}^{(m)}\,\Big|\,\zeta^{(m)}_{i+1}=t+1\Big)\\
        &+\pp\Big(\zeta^{(m)}_{i+1}=t-1\,\Big|\,\zeta^{(m)}_i=t\Big)\Big(\kappa_{i+1}^{(m)}\,\Big|\,\zeta^{(m)}_{i+1}=t-1\Big)\\
        &=\E(\kappa^{(m)}_{i+1}\mid \zeta^{(m)}_i=t)
        \end{align*}
        If $t>\frac{\omega \sqrt{m}}{2}$, a similar argument can be established using a weaker recurrence property.
    \end{proof}

    \begin{corollary}
        Let $m \in \mN$ be a varying parameter. $\kappa_0^{(m)} \leq \E\Big(\kappa_{h(m)}^{(m)}\Big)$ as long as $h(m)=o_m(m)$ and $m$ is sufficiently large.
    \end{corollary}

    The last corollary allows us to compare the coefficients of $a_{m-h(m),r}$ based on the behavior of $\zeta_{h(m)}^{(m)}$. 
    Before we proceed, we need the following two theorems, the first is used to prove Theorem \ref{rmdef}, and the second is used in Remark \ref{remark}.
    \begin{theorem}[Hoeffding]\label{hoeffding}
         Let $n \in \mN$. Let $\xi_1, \xi_2 \ldots \xi_n$ be independent random variables satisfying $a_i \leq \xi_i \leq b_i,\, i \in [n]$. Define $S_n=\sum_{i=1}^n \xi_i$. The following inequalities hold for any $t>0$:
        \begin{align*}
        &\pp(S_n-\E S_n \geq t) \leq e^{\frac{-2t^2}{\sum_{i=1}^n (b_i-a_i)^2}},\\
        &\pp(S_n-\E S_n \leq -t) \leq e^{\frac{-2t^2}{\sum_{i=1}^n (b_i-a_i)^2}}. 
        \end{align*}
        
    \end{theorem}

%{\color{olive} Suggestion: add this result as part (2) of Theorem 3.1.} {\color{teal} Thanks for moving the result.}
    \begin{theorem}\label{rmdef}
        Let $\delta\in(0,\frac{1}{2}), \omega \in (0,1)$ be the fixed parameters and $\epsilon \in (0,1),\,m \in \mN$ - varying parameters. The parameter $\Tilde{r}(m,\epsilon,\omega)=\left\lfloor\frac{m}{2}+(C(\epsilon)-\omega)\sqrt{m}\right\rfloor$ satisfies $a_{m,\Tilde{r}}=2^{-\Omega_m(\log(1-\Delta_m(\epsilon))\sqrt{m})}2^m$, where $C(\epsilon)$ is defined in Corollary \ref{r_star}.
    \end{theorem}

    \begin{proof}
        Let $q_m=\frac{1-\Delta_m}{2}$. First, we analyze the behavior of $\zeta^{(m)}_{h(m)}$. Consider $\zeta_i'$ - a sum of $i$ $\mathrm{Rad}\Big(\frac{1-\Delta_m}{2}\Big)$ independent random variables, which constitutes a transient random walk (a random walk with the probability to return to the starting position in finite time smaller than 1). Define $p=\pp\Big(\exists i : \zeta_i'=1\Big)$. The following holds:
        \begin{align*}
        p \leq \sum_{i=0}^{+\infty} \pp(\zeta_{2i+1}'=1) \leq \sum_{i=0}^{+\infty}(2i+1)q_m^{i+1}=\frac{q_m}{1-q_m}+\frac{2q_m^2}{(1-q_m)^2}.
        \end{align*}
        Consequently, $p=O(q_m)$. This implies that the probability of the event $\exists i: \zeta_i'\geq \frac{\omega\sqrt{m}}{2}$ is at most $p^{\frac{\omega\sqrt{m}}{2}}=e^{-\Omega_m(\log(q_m)\sqrt{m})}$. Note that $\pp\Big(\forall i: \zeta^{(m)}_i <  \frac{\omega\sqrt{m}}{2}\Big)=\pp\Big(\forall i: \zeta'_i < \frac{\omega\sqrt{m}}{2}\Big)$. More concretely, consider the stopping times
        \begin{align*}
            &\tau=\min\Big(\Big\{i \in [h(m)]\,\Big|\, \forall j<i: \zeta^{(m)}_j<\frac{\omega \sqrt{m}}{2}; \zeta^{(m)}_i\geq \frac{\omega \sqrt{m}}{2}\Big\}\cup\{+\infty\}\Big),\\
            &\tau'=\min\Big(\Big\{i \in [h(m)]\,\Big|\, \forall j<i: \zeta'_j<\frac{\omega \sqrt{m}}{2}; \zeta'_i\geq \frac{\omega \sqrt{m}}{2}\Big\}\cup \{+\infty\}\Big).
        \end{align*}
        Consider the processes $\zeta_{1,i}=\zeta^{(m)}_{\min(\tau,i)},\zeta_{2,i}=\zeta'_{\min(\tau',i)}$. Note that $\zeta_{1,i}$ and $\zeta_{2,i}$ are processes with the same distributions, as they are both asymmetric random walks with a ceiling at $\frac{\lceil \omega \sqrt{m}\rceil}{2}$ and the same probability parameter.
        Finally, because
        \begin{align*}
            &\pp\Big(\exists i: \zeta^{(m)}_i \geq \frac{\omega\sqrt{m}}{2}\Big)=\pp\Big(\exists N \,\,\forall i>N: \zeta_{1,i}\geq \frac{\omega \sqrt{m}}{2}\Big)\\
            &=\pp\Big(\exists N \,\, \forall i>N: \zeta_{2,i}\geq \frac{\omega \sqrt{m}}{2}\Big)=\pp\Big(\exists i: \zeta'_i \geq  \frac{\omega\sqrt{m}}{2}\Big)\le p^{\frac{\omega\sqrt{m}}{2}}
        \end{align*}
        Draw $\Tilde{\zeta}_i$ from the probability distribution of $\zeta_i$ conditioned on $\max_{i \in [h(m)]} \zeta_i \leq \frac{\omega\sqrt{m}}{2}$ - a random walk with a ceiling. Note that there is a coupling between $\zeta_i'$ and $\Tilde{\zeta}_i$ observing that the probability distribution of $\Tilde{\zeta}_i$ is also the probability distribution of $\zeta_i'$ conditioned on $\max_{i \in [h(m)]} \zeta'_i \leq \frac{\omega\sqrt{m}}{2}$. This, in turn, shows that
        $\pp\big(\zeta'_i \geq t\big) \geq  \pp\big(\Tilde{\zeta}_i \geq t\big)$ for all $t$ and $i$.

        Note that $\E \zeta'_{h(m)}=\Delta_m h(m)$, as such the Hoeffding inequality implies
        \[\pp\Big(\Tilde{\zeta}_{h(m)} \geq \frac{-h(m)\Delta_m }{2}\Big) \leq \pp\Big(\zeta'_{h(m)} \geq \frac{-h(m)\Delta_m}{2}\Big) \leq e^{-\frac{h(m)\Delta_m^2}{2}}=e^{-\Omega_m(h(m))}.\]
        Now we will establish bounds on $\E\Big(\kappa_{h(m)}^{(m)}\Big)$.

        Consider the following expression:
        \begin{align*}
        &\E\Big(\kappa_{h(m)}^{(m)}\Big)=\E\Big(\kappa_{h(m)}^{(m)}\,\Big|\,\max_i \zeta_i>\frac{\omega \sqrt{m}}{2}\Big)\pp\Big(\max_i \zeta_i>\frac{\omega \sqrt{m}}{2}\Big)\\
        &+\E\Big(\kappa_{h(m)}^{(m)}\,\Big|\,-\frac{h(m)\Delta_m}{2} < \zeta_{h(m)}, \max_i \zeta_i\leq \frac{\omega \sqrt{m}}{2}\Big)\\
        &\cdot\pp\Big(-\frac{h(m)\Delta_m}{2} < \zeta_{h(m)}, \max_i \zeta_i\leq \frac{\omega \sqrt{m}}{2}\Big)\\
        &+\E\Big(\kappa_{h(m)}^{(m)}\,\Big|\,\zeta_{h(m)}\leq -\frac{h(m)\Delta_m}{2},\max_i \zeta_i\leq \frac{\omega \sqrt{m}}{2}\Big)\\
        &\cdot\pp\Big(\zeta_{h(m)}\leq -\frac{h(m)\Delta_m}{2},\max_i \zeta_i\leq \frac{\omega \sqrt{m}}{2}\Big).
        \end{align*}
        As $\kappa_{h(m)}^{(m)}<1$, the first two summands are bounded by
        \begin{align*}
        &\pp(\max_i \zeta_i>\frac{\omega \sqrt{m}}{2})+\pp(-\frac{h(m)\Delta_m}{2} < \zeta_{h(m)}, \max_i \zeta_i \leq \frac{\omega \sqrt{m}}{2})\\
        &=\pp(\max_i \zeta_i>\frac{\omega \sqrt{m}}{2})+\pp(\frac{h(m)\Delta_m}{2} < \zeta_{h(m)}\mid \max_i \zeta_i \leq \frac{\omega \sqrt{m}}{2})\pp(\max_i \zeta_i \leq \frac{\omega \sqrt{m}}{2}) \\
        &\leq \pp(\max_i \zeta_i>\frac{\omega \sqrt{m}}{2})+\pp(-\frac{h(m)\Delta_m}{2} < \zeta_{h(m)}\mid \max_i \zeta_i \leq \frac{\omega \sqrt{m}}{2}) \\
        &\leq p^{\frac{\omega \sqrt{m}}{2}}+e^{-\frac{h(m)\Delta_m^2}{2}}=e^{-\Omega_m(\log(q_m)\sqrt{m})}+e^{-\Omega_m(h(m))}.
        \end{align*}
        The third term is bounded using the inequality
        \begin{align*}
        &a_{m-h(m),\left\lfloor m/2+(C(\epsilon)-\omega)\sqrt{m}\right\rfloor+(\zeta^{(m)}_{h(m)}-h(m))/2} \\
        &\leq \binom{m-h(m)}{\leq \left\lfloor\frac{m}{2}+(C(\epsilon)-\omega)\sqrt{m}\right\rfloor+\frac{\zeta^{(m)}_{h(m)}-h(m)}{2}}.
        \end{align*}
        Conditioned on $\zeta_{h(m)}\leq \frac{-h(m)\Delta_m}{2}$, the upper bound for the binomial coefficient is
        \[\binom{m-h(m)}{\leq \left\lfloor\frac{m}{2}+(C(\epsilon)-\omega)\sqrt{m}\right\rfloor-\left(\frac{1}{2}+\frac{\Delta_m}{4}\right)h(m)} \leq 2^{m-h(m)}e^{\frac{-\Omega_m(h(m))^2}{m-h(m)}}.\]
        if $h(m)=\omega(\sqrt{m})$.
        The upper bound is obtained from Hoeffding's inequality applied to a sum of $m-h(m)$ iid $Ber(1/2)$ random variables, where $t=\frac{m-h(m)}{2}-r$ with $r=\Omega_m(h(m))$ if $\lim_{m \rightarrow +\infty}\frac{h(m)}{\sqrt{m}}=+\infty$. Altogether, $\E \kappa_{h(m)}^{(m)} \leq e^{-\Omega_m(\log(q_m)\sqrt{m})}+e^{-\Omega_m(h(m))}+e^{-\Omega_m(\frac{h(m)^2}{m})}$, giving us the bound of $e^{-\Omega_m(\log(q_m)\sqrt{m})}$ for $h(m)=\lceil m^{\frac{3}{4}}\sqrt{\log(q_m)} \rceil$. Due to Corollary 7.4, \[\kappa_0^{(m)}=\frac{a_{m,\lfloor m/2+(C(\epsilon)-\omega)\sqrt{m} \rfloor}}{2^m} \leq \E \kappa_{h(m)}^{(m)} \leq e^{-\Omega_m(\log(q_m)\sqrt{m})}\]
        for sufficiently large $m$. Finally, one can set $\tilde{r}(m,\epsilon, \omega)=\left\lfloor m/2+(C(\epsilon)-\omega)\sqrt{m}\right\rfloor$.

    \end{proof}
    \begin{theorem}[Doob]
         Let $n \in \mN$. Let $X_1, X_2 \ldots X_n$ be a discrete-time submartingale with respect to its natural filtration. The following inequality holds: 
        \[\pp(\max_{i \in [n]} X_i \geq C) \leq \frac{\E \max(X_n,0)}{C}.\]
    \end{theorem}

    \begin{remark}\label{remark}
        Let $m \in \mN$ be a varying parameter. We provide an argument here that the bound on $a_{m,\Tilde{r}}$ cannot be improved with the same restrictions using the same approach. Note that
        \[\pp\left(\zeta_{\left\lceil\left(\frac{\omega}{2}+2\right)\sqrt{m}\right\rceil}=\left\lceil\left(\frac{\omega}{2}+2\right)\sqrt{m}\right\rceil\right)=e^{-\theta_m(\sqrt{m})}.\]
        Consider a symmetric random walk, $S_t$ with $S_0=0$. As $S_t$ is a martingale, Doob's inequality can be used. It gives
        \[\pp(\max_{i \in [h(m)-\left\lceil\left(\frac{\omega}{2}+2\right)\sqrt{m}\right\rceil]}S_i \geq 2\sqrt{m}) \leq \frac{\E \max(S_{h(m)-\left\lceil\left(\frac{\omega}{2}+2\right)\sqrt{m}\right\rceil},0)}{2\sqrt{m}}.\]
        Note that $(\E \max(S_t,0))^2 \leq (\E |S_t|)^2 \leq Var(S_t)=\sum_i Var(\xi_i)=t$. As such,
        \[\pp(\max_{i \in [h(m)-\left\lceil\left(\frac{\omega}{2}+2\right)\sqrt{m}\right\rceil]}S_i \geq 2\sqrt{m}) \leq \frac{\E \max(S_{h(m)-\left\lceil\left(\frac{\omega}{2}+2\right)\sqrt{m}\right\rceil},0)}{2\sqrt{m}} \leq \frac{1}{2}.\]

        So, with $e^{-\theta_m(\sqrt{m})}$ probability, $\zeta^{(m)}_i$ passes the $\frac{\omega \sqrt{m}}{2}$ threshold and stays above it, therefore $\zeta^{(m)}_{h(m)} \geq \frac{\omega \sqrt{m}}{2} $ with probability at least $e^{-\theta_m(\sqrt{m})}$. But conditioned on $\zeta^{(m)}_{h(m)} \geq \frac{\omega \sqrt{m}}{2}$, the trivial upper bound for $\kappa^{(m)}_{h(m)}$ is $\theta_m(1)$, which is not enough to improve upon the $e^{-\Omega_m(\sqrt{m})}$ threshold.
    \end{remark}
    
    \begin{corollary}
        Let $\delta\in(0,\frac{1}{2})$ be a fixed parameter and $m \in \mN,\,\epsilon, \omega \in (0,1)$ - varying parameters. The parameter $\Tilde{r}(m,\epsilon,\omega)=\left\lfloor\frac{m}{2}+(C(\epsilon)-\omega)\sqrt{m}\right\rfloor$ satisfies $a_{m,\Tilde{r}}=2^{m-\Omega_m(\sqrt{m})}$.
    \end{corollary}
\subsection{Proof of Theorem 3.1 (3): A formal algorithm to link the entropy bound to bit-error probability}\label{Sec7}\label{Sec7.2}
    For this section, introduce the following notation. Denote $W \subseteq [n],\, A \in \F_2^{n}$
    \begin{align*}
    &A_W=\mathrm{proj}_{[n],W}(A),\, w_W(A)=|\{i\in W\mid A_i=1\}|,\, w(A)=w_{[n]}(A).
    \end{align*}
    Note that $w(\cdot)$ is the Hamming weight. First, we prove the list decoding property.
    \begin{lemma}\label{listdec}
        Let $\mathcal{C} \subseteq \F_2^n$ be a linear code, $c\in (1,+\infty),\,\delta \in [0, \frac{1}{2})$. Consider $X \sim Unif(\mathcal{C}),\,Z \sim Ber(\delta)^{\F_2^m},\, Y=X+Z$. One can construct a list of codeword candidates $L_Y$ of cardinality $2^{cH(X \mid Y)}$ which, with probability at least $\left(1-\sqrt{\frac{1}{c}}\right)^2$, contains the true codeword.
    \end{lemma}
    \begin{proof}
        Let $y \in \F_2^n$ be an instance of $Y$. Define the following:
        \begin{itemize}
            \item $p=2^{-cH(X \mid Y)}$ is the probability threshold parameter,
            \item $V_y=(X\mid Y=y)$ is the random variable defined by the probability distribution $\forall x \in \mathcal{C}: \pp(V_y=x)=\pp(X=x \mid Y=y)$,
            \item $S_y=\{v \in \mathcal{C} \mid \pp(V_y=v)>p\}$ is a set defined by $y$ representing the most likely values of $X$ which appear with probability above the threshold $p$,
            \item $\xi_y=H(V_y)$ is the entropy of $V_y$ depending on the instance $y$,
            \item $A_y=\{\xi_y<cH(X|Y)\}$ represents an event that the entropy $H(X \mid Y=y)$ is bounded by $c\E_yH(X \mid Y=y)=cH(X|Y)$, 
            \item $B_y=\{V_y \in S_y\}$ is the event that observing $y$, the codeword $X$ is in the set $S_y$ of most likely decoding candidates.
            \item $S_y'=(S_y\mid A_y)$ is the random variable that represents the set $S_y$ when $A_y$ is true. $S_y'$ is defined by the probability distribution $\forall S \in 2^{\mathcal{C}}:\pp(S_y'=S)=\pp(S_y=S \mid A_y)$,
            \item $V_y'=(V_y \mid A_y)$ is the random variable that represents the set $V_y$ when $A_y$ is true. $V_y'$ is defined by the probability distribution $\forall x \in \mathcal{C}:\pp(V_y'=x)=\pp(V_y=x \mid A_y)$.
        \end{itemize}
\begin{comment}
        Consider events $A=\{\xi<H(X|Y)c(n)\},\, B=\{V \in S\}$, where 
        \begin{equation}\label{rv_def}
        \begin{split}
        &\xi=H(X\mid Y=y),\\
        &V=(X\mid Y=y),\\
        &S=\{v\mid \pp(V=v)>p\},\\
        &p=2^{-H(X \mid Y)c(n)}.
        \end{split}
        \end{equation}
        $\pp(A)$ represents the probability that the entropy given an observation $y$ is small, and $\pp(B\mid A)$ represents the probability that the true codeword is in the list $S$ given small entropy. $S$ is a set defined by $y$ representing the most likely values of $X$ which appear with probability above the threshold $p$.
\end{comment}
We split the proof into 4 major parts.

        \textbf{Objective}.
        We can construct the list $L$ if: 
        \begin{itemize}
            \item Both events $A_y$ and $B_y$ hold true. This happens with probability $\pp(A,B)$. For the Theorem to hold true, it is sufficient to prove $\pp(A,B)=\pp(A)\pp(B\mid A)\geq \left(1-\sqrt{\frac{1}{c}}\right)^2.$
            \item $|S_A| \leq 2^{cH(X|Y)}$. Then, the list $L$ can be compiled by collecting the $2^{cH(X|Y)}$ most likely codewords.
        \end{itemize}
        
        \textbf{Bounding $\pp(A_y)$}. This part is used to compute the lower bound for $\pp(A_y,B_y)$. The Markov inequality implies:
        \[\pp(\xi_y \geq t\E_y\xi_y)=\pp(\xi_y \geq tH(X \mid Y)) \leq \frac{1}{t}.\]
        Taking $t=\sqrt{c}$, we conclude that 
        \[\pp\left(\xi_y < H(X \mid Y)\sqrt{c}\right) \geq 1-\sqrt{\frac{1}{c}}.\]

        \textbf{Bounding $\pp(B_y\mid A_y)$}. This part both establishes a good lower bound on $\pp(A_y,B_y)$, as well as shows an upper bound on $|S_y|$. Consider the following statement:

        \textsl{Let V' be a $\mathcal{V}$-valued random variable, $p>0$, $S'=\{v \in \mathcal{V}\mid \pp(V'=v)>p\}$. Then $|S'|\le \frac{1}{p}$ and $H(V')\geq (1-\pp(V' \in S'))\log_2 \frac{1}{p}$}

        \textsl{Proof:}
        \begin{enumerate}
            \item $1 \geq \sum_{v \in S'}\pp(V'=v) \ge p|S'|$
            \item $H(V')=-\E \log_2 \pp(V'=v) \geq -\E \log_2 \pp(V'=v) \mathbbm{1}_{v \notin S'} \geq \pp(V' \notin S') \log_2 \frac{1}{p}$
        \end{enumerate}
        Consider the statement for the random variable $V'=V_y'$, for which $S'=S_y'$. The following is implied:

        \begin{itemize}
            \item$|S_y'|\le 1/p=2^{cH(X \mid Y)},$
            \item$H(X \mid Y)\sqrt{c} \geq H(V_y') \geq (1-\pp(V_y' \in S_y'))H(X \mid Y)c$
            \item Therefore, $\pp(B_y|A_y)=\pp(V_y' \in S_y') \geq 1-\sqrt{\frac{1}{c}}.$
        \end{itemize}
        %\textbf{Explaining the bound on $\pp(B)$}. By taking at least the $\frac{1}{p}$ closest codewords, we ensure that the set $S$ is contained in the composed list, implying the statement of the Corollary. 

    \textbf{Conclusion}. Given event $A_y$, we have constructed a set $S_y'$ satisfying $|S_y'| \leq 2^{cH(X \mid Y)}$. Assuming that the event $A_y$ holds, $S_y'$ contains the true codeword if and only if the event $B_y$ holds. Note: $\pp(A_y, B_y)=\pp(A_y)\pp(B_y|A_y)\geq\left(1-\sqrt{\frac{1}{c}}\right)^2.$
    We have satisfied both conditions and thus have proven the Theorem.
    \end{proof}
    \begin{corollary}\label{list}
        Let $m \in \mN,\, \epsilon>0$ be varying parameters, $c>0,\,\delta \in (0, \frac{1}{2}), \omega>0$. Assume that the noisy $RM(m,\Tilde{r}(m,\epsilon, \omega))$ codeword is considered, where $\Tilde{r}(m,\epsilon, \omega)$ is defined in Theorem \ref{rmdef}. One can construct a list of codeword candidates $L$ of cardinality $2^{a_{m,\Tilde{r}}2^{c(1-\Delta_m(\epsilon))\sqrt{m}}}$ which, with probability $1-2^{-\Omega_m((1-\Delta_m(\epsilon))\sqrt{m})}$, contains the true codeword.
    \end{corollary}
    We need the following property in this section.
    \begin{property}
     Let $n\in \mN,\,i \in [n]$. Let $P_{\mathrm{bit},i}:=\pp(\widehat{X_i}(Y) \neq X_i)$, where $\widehat{X_i}$ denotes the most likely value of $X_i$ given $Y$. Reed-Muller codes' associated bit-error probabilities satisfy  the following: $\forall S\in [n]: P_{\mathrm{bit},i}=P_{\mathrm{bit}}$.
    \end{property}
    
    Let $\delta \in \left(0,\frac{1}{2}\right)$.
    We introduce and analyze a formal decoding algorithm in this section.\\

    \fbox{\begin{minipage}{33em}
    \textbf{Algorithm:}
    \begin{enumerate}
    \item Choose $\delta'>\delta$ such that the rate of the code is less than $1-H(\delta')$.

    \item Call the noisy codeword $Y'$.

    \item Define a set $S$ that contains each element of $\mathbb{F}_2^m$ with probability $\gamma=\frac{2(\delta'-\delta)}{1-2\delta}$.

    \item Create a new corrupted codeword $Y''$ by setting $Y''_i=Y'_i$ for $i\not\in S$ and setting $Y''_i \sim Ber(1/2)$ randomly for $i\in S$.

    \item Use the new corrupted codeword to make a list of $2^{e^{-\Omega_m(\sqrt{m})} 2^{m}}$ codewords that almost certainly contains the true codeword.

    \item Return the codeword from the list that agrees with the original corrupted codeword on the most bits in $S$.
    \end{enumerate}
    \end{minipage}}\\
    \vspace{0.3pt}\\
    The proof of the main theorem relies on the following important property:
\begin{comment}
    \begin{property}
        Let $n \in \mN$. Consider two error patterns $S,\,T$ valued in $2^{[n]}$ which are induced by passing the codeword through two BSC channels with error parameters $\gamma,\, \delta$ subsequently. Consider the error pattern $S \Delta T$ which is achieved by flipping the bits in sets $S$ and $T$ consequently. Then, the following statements are true:
        \begin{itemize}
            \item $S \Delta T$ is an error pattern distributed as the error induced by the BSC($\delta+\gamma-2\gamma\delta$) channel.
            \item The indicators $\{\mathbbm{1}_{i \in S}\}_{i \in [n]}$ are mutually independent conditioned on a realization of $S \Delta T$ and have an expectation of at least $C'_{\delta, \gamma}>0$, where $C'_{\delta,\gamma}$ is independent of $n$ and the realization of $S \Delta T$.
        \end{itemize}
        
    \end{property}
\end{comment}
\begin{property}
    Let $m \in \mN,\,\gamma,\delta \in (0,\frac{1}{2})$. Consider two error vectors $X \sim Ber(\gamma)^{n},\, Y \sim Ber(\delta)^{n}$. Then, the following holds true:
    $X+Y \sim Ber(\delta+\gamma-2\gamma\delta)^{n}$.
    Moreover, for any $z \in \F_2^{n},$ the random bits $\{X_i\mid X+Y=z\}_{i \in [n]}$ are mutually independent. Finally, $\pp(X_i=1 \mid X+Y=z) \geq C_{\delta, \gamma}'$, where $C'_{\delta,\gamma}>0$ is independent of $m$ and $z$.  
\end{property}
\begin{proof}
    Consider the realization $X+Y=\1_U,\, U \subseteq [n],\,W \subseteq [n]$. Note that
    \begin{align*}
    &\pp(X_W=\1, X+Y=\1_U)=\sum_{S_1 \subseteq [n]\setminus W}\pp(X=\1_{S_1 \cup W}, Y = \1_{U \Delta (S_1 \cup W)})\\
    &=\sum_{S_1 \subseteq [n]\setminus W}\pp(X=\1_{S_1 \cup W}, Y = \1_{(U \Delta S_1)\cup (W \setminus U) \setminus (W \cap U)})\\
    &=\sum_{S_1 \subseteq [n]\setminus W} \Big(\gamma^{|S_1|+ |W|}(1-\gamma)^{n-|S_1|-|W|} \\
    &\cdot\delta^{|U \Delta S_1|+|W \setminus U|-|W \cap U|} (1-\delta)^{n-(|U \Delta S_1|+|W \setminus U|-|W \cap U|)}\Big)\\
    &=\left(\frac{\gamma(1-\delta)}{\delta(1-\gamma)}\right)^{|W \cap U|}\left(\frac{\gamma \delta}{(1-\gamma)(1-\delta)}\right)^{|W \setminus U|}\\
    &\cdot\sum_{S_1 \subseteq [n]\setminus W} \gamma^{|S_1|}(1-\gamma)^{n-|S_1|}\delta^{|U \Delta S_1|}(1-\delta)^{n-|U \Delta S_1|};\end{align*}
    \begin{align*}
    &\pp(X+Y = \1_U)=\sum_{S_1 \subseteq [n]}\pp(X=\1_{S_1},Y=\1_{U \Delta S_1})\\
    &=\sum_{S_1 \subseteq [n]} \gamma^{|S_1|}(1-\gamma)^{n-|S_1|}\delta^{|U \Delta S_1|}(1-\delta)^{n-|U \Delta S_1|},
    \end{align*}
    for all $W$ and $U$. Note that $\sum_{S_1 \subseteq [n]\setminus W} \gamma^{|S_1|}(1-\gamma)^{n-|S_1|}\delta^{|U \Delta S_1|}(1-\delta)^{n-|U \Delta S_1|}$ and the same sum over $[n]\setminus W \cup \{u\}$ differ by a factor of $1+\frac{\gamma \delta}{(1-\gamma)(1-\delta)}$ if $u \in U^c$, as $|S_1 \cup \{u\}|=|S_1|+1$ and $ |U\Delta(S_1 \cup \{u\})|=|U\Delta S_1|+1$. Analogously, the sums differ by a factor of $1+\frac{\gamma(1-\delta)}{\delta(1-\gamma)}$ if $u \in U$. As such,
    \begin{align*}
    &\pp(X_W=\1 \mid X+Y = \1_U)\\
    &=\left(\frac{\gamma(1-\delta)}{\gamma(1-\delta)+\delta(1-\gamma)}\right)^{|W \cap U|}\left(\frac{\gamma \delta}{\gamma \delta+(1-\gamma)(1-\delta)}\right)^{|W \setminus U|}.
    \end{align*}
    Consequently, 
    \begin{align*}
    &\pp\big(X_i=1 \mid  X+Y = \1_U\big)\\
    &=\1_{i \in U} \frac{\gamma(1-\delta)}{\gamma(1-\delta)+\delta(1-\gamma)} + \1_{i \in U^c} \frac{\gamma \delta}{\gamma \delta+(1-\gamma)(1-\delta)}.
    \end{align*}
    
    One can use $W=[n]$ to show:
    \[\pp\big(X=\1\,\big|\,X+Y= \1_U\big)=\left(\frac{\gamma(1-\delta)}{\gamma(1-\delta)+\delta(1-\gamma)}\right)^{|U|}\left(\frac{\gamma \delta}{\gamma \delta+(1-\gamma)(1-\delta)}\right)^{n-|U|}.\]
    The following is true:
    \begin{align*}
     &\pp(X=\1_W, X+Y=\1_U)=\pp(X=\1_W, Y=\1_{U\Delta W})\\
     &=\gamma^{|W|}(1-\gamma)^{n-|W|}\delta^{|U \Delta W|}(1-\delta)^{n-|U \Delta W|}.
    \end{align*}
    As such, the previous equality and this relation for $W=[n]$ imply that
    \begin{align*}
        &\pp\big(X+Y=\1_U\big)=\frac{\pp\big(X=\1, X+Y=\1_U\big)}{\pp\big(X=\1\mid X+Y=\1_U\big)}\\
        &=\gamma^n \delta^{n-|U|}(1-\delta)^{|U|}\Big/\left(\left(\frac{\gamma(1-\delta)}{\gamma(1-\delta)+\delta(1-\gamma)}\right)^{|U|}\left(\frac{\gamma \delta}{\gamma \delta+(1-\gamma)(1-\delta)}\right)^{n-|U|}\right)\\
        &=(\gamma(1-\delta)+\delta(1-\gamma))^{|U|}(\gamma\delta+(1-\gamma)(1-\delta))^{n-|U|}\\
        &=(\gamma+\delta-2\gamma\delta)^{|U|}(1-(\gamma+\delta-2\gamma\delta))^{n-|U|}.
    \end{align*}
    Thus, $X+Y \sim Ber(\delta+\gamma-2\delta\gamma)$. Combining all these relations, we obtain:
    \begin{align*}
    &\pp\big(X=\1_W\big|X+Y=\1_U\big)%=\pp\Big((\mathbbm{1}_{i \in S})_{i \in W}=1^{|W|},(\mathbbm{1}_{i \in S})_{i \in W^c}=0^{|W^c|}\Big|S\Delta T=U\Big)\\
    =\frac{\pp(X=\1_W,X+Y=\1_U)}{\pp(X+Y=\1_U)}\\
    &=\frac{\gamma^{|W|}(1-\gamma)^{n-|W|}\delta^{|U \Delta W|}(1-\delta)^{n-|U \Delta W|}}{(\gamma(1-\delta)+\delta(1-\gamma))^{|U|}(\gamma\delta+(1-\gamma)(1-\delta))^{n-|U|}}\\
    &=\left(\frac{\gamma(1-\delta)}{\gamma(1-\delta)+\delta(1-\gamma)}\right)^{|W\cap U|} \left(\frac{\gamma\delta}{\gamma\delta+(1-\gamma)(1-\delta)}\right)^{|W\setminus U|}\\
    &\cdot \left(\frac{\delta(1-\gamma)}{\gamma(1-\delta)+\delta(1-\gamma)}\right)^{|W^c \cap U|}\left(\frac{(1-\gamma)(1-\delta)}{\gamma\delta+(1-\gamma)(1-\delta)}\right)^{|W^c \setminus U|}.
    \end{align*}
    This relation implies the mutual independence of bits $\{X_i|X+Y=z\}_{i \in [n]}$. Finally, $\pp(X_i=1|X+Y=z)$ is bounded from below by $\min(\frac{\gamma\delta}{\gamma\delta+(1-\gamma)(1-\delta)},\frac{\gamma(1-\delta)}{\gamma(1-\delta)+\delta(1-\gamma)})$, which only depend on parameters $\delta$ and $\gamma$, but not characteristics of $U$.
    \end{proof}
    We are ready to state the final theorem.
    \begin{theorem}
        Consider the binary symmetric channel with error parameter $\delta \in [0,\frac{1}{2})$. Consider a family of codes $\{\mathcal{C}_i\}_{i \in \mN}$ of length $n_i$ satisfying two properties:
        \begin{itemize}
            \item Let $
             X^{(i)} \sim Unif(\mathcal{C}_{i}),\, Y'^{(i)}=X^{(i)}+Z,\,Z \sim Ber(\delta)^{n_i}$. The following is satisfied: 
            \[d(\mathcal{C}_i, \delta):=\frac{H(X^{(i)} \mid Y'^{(i)})}{n_i}=o_{n_i}(1).\] 
            \item The code-associated bit-error probability satisfies $\forall j \in [n_i]: P_{\mathrm{bit},j}=P_{\mathrm{bit}}$.
        \end{itemize}
        Then $P_{\mathrm{bit}}=O_{n_i}\left(d(\mathcal{C}_i, \delta)^{1/3}\right)$.
    \end{theorem}
\begin{proof}
    Consider $Y''$ - a noisy version of $Y'$ with set $S$ from the algorithm, $X+Y'=Z,\,\, X+Y''=Z'$ and $\Lambda$ - a codeword in $\mathcal{C}_j$ that depends on $X,Y''$ and is conditionally independent from $Z_S$ and $S$ given $X,Y''$.
    For brevity, let $\pp(E\mid X, Y'')$ be the probability of an event $E$ conditioned on true codeword with instance $X$ and noisy codeword with instance $Y''$. Note that $Y''$ and $Z_S$ are independent conditioned on $S$. As such, for any set $W\subseteq S$ and $W'=S\backslash W$
    \[\pp(Y'_W=X_W, Y'_{W'} \neq X_{W'}\mid S,X,Y'')=(1-\delta)^{|W|}\delta^{|W'|}.\,\,\]
    That means that if $T=\{i: X_i\ne \Lambda_i\}$, the following is true:
    \begin{align*}
     &\pp(w_S(X+Y') \geq w_S(\Lambda+Y')\mid S,X,Y'')\\
     &=\pp(w_S(Z) \geq w_S((\Lambda+X)+Z)\mid S,X,Y'')\\
    &=\pp(w_{S\cap T}(Z) \geq w_{S \cap T}((\Lambda+X)+Z)\mid S,X, Y'')\\
    &=\sum_{i=0}^{\frac{|S \cap T|}{2}}\pp(w_{S \cap T} ((\Lambda+X)+Z)=i\mid S,X,Y'')\\
    &=\sum_{i=0}^{\frac{|S \cap T|}{2}}\binom{|S\cap T|}{i}\delta^{|S \cap T| - i}(1-\delta)^{i} \leq (4 \delta (1-\delta))^{\frac{|S \cap T|}{2}}.
    \end{align*}
    The final bound comes from the observation that
    \[\sum_{i=0}^{\frac{|S \cap T|}{2}}\binom{|S\cap T|}{i}\left(\frac{\delta}{1-\delta}\right)^{\frac{|S \cap T|}{2}-i} \leq \sum_{i=0}^{\frac{|S \cap T|}{2}}\binom{|S\cap T|}{i} = 2^{|S \cap T|-1} \leq 4^{\frac{|S \cap T|}{2}}.\]
    Note that
    \begin{align*}
    &\pp(w_S(Z) \geq w_S((\Lambda-X)-Z)\mid X,Y'')\\
    &=\sum_S\pp(w_S(Z) \geq w_S((\Lambda-X)-Z)\mid S,X,Y'')\pp(S\mid X,Y'').
    \end{align*}
    Denote $S' \subseteq S$ to be the induced error pattern
    ($S'=\{i \in [n]\mid (Z'+Z)_i=1\}$). Note the following:
    \begin{align*}
     &\pp(w_S(X+Y') \geq w_S(\Lambda+Y')\mid X,Y'')\\
     &=\E_{S}[\pp(w_S(X+Y') \geq w_S(\Lambda+Y')\mid S,X,Y'')\mid X,Y'']\\
     &\leq \sum_U \pp(S=U\mid X,Y'')(4 \delta (1-\delta))^{\frac{|U \cap T|}{2}} \\
     &=\sum_U\sum_V \pp(S=U, S'=V\mid X,Y'')(4 \delta (1-\delta))^{\frac{|U \cap T|}{2}} \\
     & \leq \sum_V \sum_U \pp(S'=V,S=U\mid X,Y'')(4 \delta (1-\delta))^{\frac{|V \cap T|}{2}}\\
     &=\sum_V (4 \delta (1-\delta))^{\frac{|V \cap T|}{2}}\pp(S'=V\mid X,Y'')\\
     &=\E_{S'} \left[(4 \delta (1-\delta))^{\frac{|S' \cap T|}{2}}\,\Big|\,X,Y''\right].
     \end{align*}
     To bound the last expectation, use the fact from the last proposition that the random bits $\{X_i \mid Z'=z\}_{i \in [n]}$ are independent with $\pp(X_i=1 \mid Z'=z) \geq 2C_{\delta, \gamma}=C'_{\delta,\gamma/2}=\theta_m(1)$, where $\gamma=\frac{2(\delta'-\delta)}{1-2\delta}$. Moreover, $X$ is independent from both $Z'$ and $Z$. Note that $|S' \cap T|=\sum_{t \in T}(Z+Z')_t$. Conclude that 
     \[\E\left[\frac{|S' \cap T|}{2}\,\bigg|\,X,Z'\right]=\frac{1}{2}\E\left[\sum_{t \in T}(Z+Z')_t \,\,\bigg|\,\, Z'\right] \geq C_{\delta, \gamma}|T|.\]
     By Hoeffding's inequality,
     \begin{align*}
    &\pp\left(\frac{|S' \cap T|}{2} \leq (C_{\delta, \gamma} - \sigma)|T|\,\Big|\,X,Y''\right)=\pp\left(\frac{|S' \cap T|}{2} \leq (C_{\delta, \gamma} - \sigma)|T|\,\Big|\,X,Z'\right) \\
    &\leq e^{\frac{-2\sigma^2|T|^2}{|T|}}=e^{-2\sigma^2|T|}.
    \end{align*}
    Taking $\sigma=\frac{C_{\delta, \gamma}}{2}$, we see that $\pp\left(\frac{|S' \cap T|}{2} \geq \frac{C_{\delta, \gamma}}{2}|T|\,\Big|\,X,Y''\right) \geq 1-e^{\frac{-C_{\delta, \gamma}^2|T|}{2}}$. This implies
    \[\E_{S'}\left[(4 \delta (1-\delta))^{\frac{|S' \cap T|}{2}}\,\Big|\,X,Y''\right] \leq e^{\frac{-C_{\delta, \gamma}^2|T|}{2}}+\Bigg(1-e^{\frac{-C_{\delta, \gamma}^2|T|}{2}}\Bigg)e^{\log(4\delta(1-\delta))\frac{C_{\delta, \gamma}}{2}|T|}.\]
    RHS can be interpreted as $e^{-\Omega_n(|T|)}.$ As such, $\pp(w_S(X-Y') \geq w_S(\Lambda-Y')\mid X,Y'') \leq e^{-C'|T|}$ for some $C'>0$ and large enough $m$, and by extension $|T|$. This implies
    \[\pp\left(\bigcup_{\Lambda \in L} w_S(X-Y') \geq w_S(\Lambda-Y')\,\Big|\,X,Y''\right) \leq |L|e^{-C'n_i d(\mathcal{C}_i, \delta)^{1/3}}.\]
    Here, $L$ is the list from the Lemma \ref{listdec} excluding the elements that differ from the original codeword by at most $n_id(\mathcal{C}_i, \delta)^{1/3}$ bits. As $|L| \leq 2^{\frac{C'}{2}n_id(\mathcal{C}_i, \delta)^{1/3}}$ for $c=\frac{2}{ C'd(\mathcal{C}_i, \delta)^{2/3}}$, $|L|e^{-C'n_id(\mathcal{C}_i, \delta)^{1/3}}=o_{n_i}\left(d(\mathcal{C}_i, \delta)^{1/3}\right)$.
    
    %As $|L| \leq 2^{2^me^{-C_2\sqrt{m}}}$ for some constant $C_2$, if one picks $2C_1=C_2$, $|L|e^{-C2^me^{-C_1\sqrt{m}}} \leq e^{-2^me^{-\Omega_m(\sqrt{m})}}$ for a large enough $m$.

    Finally, we compute the expected cardinality of error bits. With probability $O_{n_i}\left(d(\mathcal{C}_i, \delta)^{1/3}\right)$, the list does not contain the true codeword. In this case, the cardinality of error bits is at most $n_iO_{n_i}\left(d(\mathcal{C}_i, \delta)^{1/3}\right)$, and the respective contribution to the expectation is $n_i$. With probability $o_{n_i}\left(d(\mathcal{C}_i, \delta)^{1/3}\right)$, there exists a codeword that is more than $n_id(\mathcal{C}_i, \delta)^{1/3}$ bits away from the true codeword and is closer to $Y''$ on $S$ than the true codeword, so the respective contribution to the expectation is at most $n_io_{n_i}\left(d(\mathcal{C}_i, \delta)^{1/3}\right)$.

    Finally, in the remaining case, a codeword that is at most $n_id(\mathcal{C}_i, \delta)^{1/3}$ bits away from the true codeword is output, thus the respective contribution to the expectation is at most $n_id(\mathcal{C}_i, \delta)^{1/3}$. Overall, the expected cardinality of error bits is $O_{n_i}\left(d(\mathcal{C}_i, \delta)^{1/3}\right)$. The following inequalities are true:

    \[P_{\mathrm{bit}}=\frac{\sum_{i=1}^{n_i} P_{\mathrm{bit},i}}{n_i} \leq \frac{\E|\{i\mid i\text{-th bit decoded incorrectly}\}|}{n_i} \leq O_{n_i}\left(d(\mathcal{C}_i, \delta)^{1/3}\right)\]
\end{proof}

\begin{corollary}
    Consider the binary symmetric channel with error parameter $\delta \in [0,\frac{1}{2})$. Assume the parameters $m$ and $r_m$ satisfy the relation $\limsup_{m \rightarrow +\infty}\frac{\binom{m}{\le r_m}}{2^m}<1-H(\delta)$, where $0 \leq r_m \leq m$. The bit-error probability of the Reed-Muller code $RM(m,r_m)$ satisfies the following relation:
        \[P_{\mathrm{bit}}=2^{-\Omega_m(\sqrt{m})}.\]
\end{corollary}
\section{Strong capacity result and new additive combinatorics conjecture}\label{new_conj}

The following conjecture would potentially be useful in strengthening our result from a rate of  $2^{-\Omega_m(\sqrt{m})}$ to  $2^{-\Omega_m(\sqrt{m} \log(m))}$, which would then also imply a vanishing block-error probability up to Shannon capacity (in the complete sense of decoding the full messages) using the bit to block results from \cite{Abbe23}.

\begin{conjecture}
For any\footnote{Note that one can focus on $c_1 \in (0,11)$ since one can otherwise use \cite{gowers2023conjecturemarton}} $c_1>0$, there exists $c_2=exp(O_m(1/c_1))$ such that for any random variable $X$ valued in $\mathbb{F}_2^m$, $X'$ an independent copy of $X$, $m \in \mZ_+$, there exists a subspace $\Gc$ of dimension at most $c_2 H(X)$ such that:
\[H(U_\Gc+X)-H(U_\Gc) \leq (1+c_1) (H(X'+X)-H(X')).\]
\end{conjecture}

\begin{remark}
  Let $\Gc$ be a subspace of $\F_2^d$, where $d \in \mN$. One can show the following:
  \[H(U_\Gc+X)-H(U_\Gc)=H(\mathrm{Proj}_{\Gc^\perp}(X)).\]
  We infer this from $-H(U_\Gc+X)=\sum_{u \in \mathbb{F}_2^m} \pp(U_\Gc+X=u) \log_2 \pp(U_\Gc+X=u)
        =\sum_{u \in \Gc^\perp} \pp(X \in \Gc+u) \log_2 \frac{\pp(X \in \Gc+u)}{|\Gc|}
        =\sum_{u \in \Gc^\perp} \pp(X \in \Gc+u) \log_2 \pp(X \in \Gc+u) - H(U_\Gc)$.
  Here, the second equality is due to the fact that
  $\pp(U_\Gc+X=u)=\sum_{u_\Gc \in \Gc}\pp(U_\Gc=u_\Gc)\pp(X=u+u_\Gc)=\frac{\sum_{u_\Gc \in \Gc}\pp(X=u+u_\Gc)}{|\Gc|}=\frac{\pp(X \in \Gc+u)}{|\Gc|}$.
  Finally, the sum over $\Gc^\perp$ is exactly $-H(\mathrm{Proj}_{\Gc^\perp}(X))$.
\end{remark}

This is a relaxation of the result of \cite{gowers2023conjecturemarton} in the sense that it requires the variability of $X$ to mostly be along $\Gc$ instead of requiring that the probability distribution of $X$ be approximately equal to $U_\Gc$. On the flip side, it is asking for tighter constants and more explicitly constrained $\Gc$.

We leave it as an open problem to establish this conjecture and close the strong capacity result using the current entropy extraction approach.
\section{Acknowledgments}
We thank Jan Hazla, Avi Wigderson, Yuval Wigderson and Min Ye for stimulating discussions on the entropy extraction approach to Reed-Muller codes, as well as Frederick Manners, Florian Richter, Tom Sanders and Terence Tao for further feedback on additive combinatorics results.  
\newpage
\bibliography{RM.bib}

@article{Arikan09,
  title={Channel polarization: {A} method for constructing capacity-achieving codes for symmetric binary-input memoryless channels},
  author={Ar{\i}kan, E.},
  journal={IEEE Transactions on Information Theory},
  volume={55},
  number={7},
  pages={3051--3073},
  year={2009},
  publisher={IEEE}
}

@article{Ye19,
  title={Reed-{M}uller codes polarize},
  author={Emmanuel Abbe and Min Ye},
  journal={IEEE Symposium on Foundations of Computer Science},
  year={2019},
  publisher={IEEE}}

@article{Hazla21,
title = "Almost-{R}eed-{M}uller Codes Achieve Constant Rates for Random Errors",
author = "Emmanuel Abbe and Jan Hazla and Ido Nachum",
note = "Publisher Copyright: {\textcopyright} 1963-2012 IEEE.",
year = "2021",
month = dec,
day = "1",
doi = "10.1109/TIT.2021.3116663",
language = "English",
volume = "67",
pages = "8034--8050",
journal = "IEEE Transactions on Information Theory",
issn = "0018-9448",
publisher = "Institute of Electrical and Electronics Engineers Inc.",
number = "12",
}

@INPROCEEDINGS{Abbe23,
  author={Abbe, Emmanuel and Sandon, Colin},
  booktitle={2023 IEEE 64th Annual Symposium on Foundations of Computer Science (FOCS)}, 
  title={A proof that {Reed-Muller} codes achieve {S}hannon capacity on symmetric channels}, 
  year={2023},
  volume={},
  number={},
  pages={177-193},
  doi={10.1109/FOCS57990.2023.00020}}

@article{DBLP:journals/jacm/DvirG16,
  author    = {Zeev Dvir and
               Sivakanth Gopi},
  title     = {2-Server {PIR} with Subpolynomial Communication},
  journal   = {J. {ACM}},
  volume    = {63},
  number    = {4},
  pages     = {39:1--39:15},
  year      = {2016}
}

@article{DBLP:journals/jacm/ChorKGS98,
  author    = {Benny Chor and
               Eyal Kushilevitz and
               Oded Goldreich and
               Madhu Sudan},
  title     = {Private Information Retrieval},
  journal   = {J. {ACM}},
  volume    = {45},
  number    = {6},
  pages     = {965--981},
  year      = {1998}
}

@article{BEIMEL2005213,
title = {General constructions for information-theoretic private information retrieval},
journal = {Journal of Computer and System Sciences},
volume = {71},
number = {2},
pages = {213-247},
year = {2005},
author = {Amos Beimel and Yuval Ishai and Eyal Kushilevitz}
}

@article{Kudekar17,
  title={Reed--{M}uller codes achieve capacity on erasure channels},
  author={Kudekar, S. and Kumar, S. and Mondelli, M. and Pfister, H. D. and {\c{S}}a{\c{s}}oǧlu, E. and Urbanke, R.},
  journal={IEEE Transactions on Information Theory},
  volume={63},
  number={7},
  pages={4298--4316},
  year={2017},
  publisher={IEEE}
}

@book{Macwilliams77,
  title={The theory of error-correcting codes},
  author={MacWilliams, F. J. and Sloane, N. J. A.},
  year={1977},
  publisher={Elsevier}
}

@article{Saptharishi17,
  title={Efficiently decoding {R}eed--{M}uller codes from random errors},
  author={Saptharishi, R. and Shpilka, A. and Volk, B. L.},
  journal={IEEE Transactions on Information Theory},
  volume={63},
  number={4},
  pages={1954--1960},
  year={2017},
  publisher={IEEE}
}

@ARTICLE{expander,
     AUTHOR= "M.~Sipser and D.~A.~Spielman",
     TITLE= "{Expander Codes}",
     JOURNAL = "IEEE Trans. on Inform. Theory",
     YEAR = "1996",
     VOLUME = "42",
     PAGES = "1710-1722",
     TYPE = "article"}

@ARTICLE{dumer1,
author={Dumer, I.},
journal={Information Theory, IEEE Transactions on},
title={{Recursive decoding and its performance for low-rate Reed-Muller codes}},
year={2004},
month={May},
volume={50},
number={5},
pages={811-823},
ISSN={0018-9448},}

@article{dumer2,
title = "{Recursive error correction for general Reed-Muller codes}",
journal = "Discrete Applied Mathematics ",
volume = "154",
number = "2",
pages = "253 - 269",
year = "2006",
note = "Coding and Cryptography ",
issn = "0166-218X",
doi = "http://dx.doi.org/10.1016/j.dam.2005.05.013",
author = "I. Dumer and K. Shabunov"
}

@ARTICLE{dumer3,
author={Dumer , I.},
journal={Information Theory, IEEE Transactions on},
title={{Soft-decision decoding of Reed-Muller codes: A simplified algorithm}},
year={2006},
month={March},
volume={52},
number={3},
pages={954-963},
ISSN={0018-9448},}

@ARTICLE{hell, 
author={Helleseth, T. and Klove, T. and Levenshtein, V. I.}, 
journal={Information Theory, IEEE Transactions on}, 
title={Error-correction capability of binary linear codes}, 
year={2005}, 
month={April}, 
volume={51}, 
number={4}, 
pages={1408-1423}, 
ISSN={0018-9448},}

@ARTICLE{arikan-RM, 
author={Arikan, E.}, 
journal={Communications Letters, IEEE}, 
title={{A performance comparison of polar codes and Reed-Muller codes}}, 
year={2008}, 
month={June}, 
volume={12}, 
number={6}, 
pages={447-449}, 
doi={10.1109/LCOMM.2008.080017}, 
ISSN={1089-7798},}

@ARTICLE{sloane-RM, 
author={Sloane, N. J. A. and Berlekamp, E.}, 
journal={Information Theory, IEEE Transactions on}, 
title={{Weight enumerator for second-order Reed-Muller codes}}, 
year={1970}, 
month={Nov}, 
volume={16}, 
number={6}, 
pages={745-751}, 
doi={10.1109/TIT.1970.1054553}, 
ISSN={0018-9448},}

@article{RM_fnt,
url = {http://dx.doi.org/10.1561/0100000123},
year = {2023},
volume = {20},
journal = {Foundations and Trends in Communications and Information Theory},
title = {Reed-{M}uller Codes},
doi = {10.1561/0100000123},
issn = {1567-2190},
number = {12},
pages = {1-156},
author = {Emmanuel Abbe and Ori Sberlo and Amir Shpilka and Min Ye}
}

@inproceedings{Abbe15stoc,
author = {Abbe, Emmanuel and Shpilka, Amir and Wigderson, Avi},
title = {Reed-{M}uller Codes for Random Erasures and Errors},
year = {2015},
isbn = {9781450335362},
publisher = {Association for Computing Machinery},
address = {New York, NY, USA},
url = {https://doi.org/10.1145/2746539.2746575},
doi = {10.1145/2746539.2746575},
booktitle = {Proceedings of the Forty-Seventh Annual ACM Symposium on Theory of Computing},
pages = {297–306},
numpages = {10},
location = {Portland, Oregon, USA},
series = {STOC '15}
}

@article{Abbe15,
  title={Reed--{M}uller codes for random erasures and errors},
  author={Abbe, E. and Shpilka, A. and Wigderson, A.},
  journal={IEEE Transactions on Information Theory},
  volume={61},
  number={10},
  pages={5229--5252},
  year={2015},
  publisher={IEEE}
}

@inproceedings{Lin93,
  title={{RM} codes are not so bad},
  author={Lin, S.},
  booktitle={IEEE Inform. Theory Workshop},
  year={1993},
  note={Invited talk}
}

@inproceedings{Dumer93,
  title={Erasure correction performance of linear block codes},
  author={Dumer, I. and Farrell, P.},
  booktitle={Workshop on Algebraic Coding},
  pages={316--326},
  year={1993},
  organization={Springer}
}

@inproceedings{Arikan2010survey,
  title={A survey of {Reed-Muller} codes from polar coding perspective},
  author={Arikan, E.},
  booktitle={2010 IEEE Information Theory Workshop on Information Theory (ITW 2010, Cairo)},
  pages={1--5},
  year={2010},
  organization={IEEE}
}

@inproceedings{Carlet05,
  title={On the construction of balanced {Boolean} functions with a good algebraic immunity},
  author={Carlet, C. and Gaborit, P.},
  booktitle={Proceedings. International Symposium on Information Theory, 2005. ISIT 2005.},
  pages={1101--1105},
  year={2005},
  organization={IEEE}
}

@misc{abbe2023reedmullercodesvanishingbiterror,
      title={{R}eed-{M}uller codes have vanishing bit-error probability below capacity: a simple tighter proof via camellia boosting}, 
      author={Emmanuel Abbe and Colin Sandon},
      year={2023},
      eprint={2312.04329},
      archivePrefix={arXiv},
      primaryClass={cs.IT},
      url={https://arxiv.org/abs/2312.04329}, 
}

@inproceedings{comparingBitMAP,
  title={Comparing the bit-MAP and block-MAP decoding thresholds of {R}eed-{M}uller codes on {BMS} channels},
  author={Kudekar, Shrinivas and Kumar, Santhosh and Mondelli, Marco and Pfister, Henry D. and Urbanke, R{\"u}diger},
  booktitle={2016 IEEE International Symposium on Information Theory (ISIT)},
  pages={1755--1759},
  year={2016}
}

@article{Mondelli14,
  title={From polar to {R}eed-{M}uller codes: {A} technique to improve the finite-length performance},
  author={Mondelli, M. and Hassani, S. H. and Urbanke, R. L.},
  journal={IEEE Transactions on Communications},
  volume={62},
  number={9},
  pages={3084--3091},
  year={2014},
  publisher={Ieee}
}

@book{Richardson08,
  title={Modern coding theory},
  author={Richardson, T. and Urbanke, R.},
  year={2008},
  publisher={Cambridge university press}
}

@inproceedings{AY18,
  title={Reed-{M}uller codes polarize},
  author={Abbe, E. and Ye, M.},
  booktitle={2019 IEEE 60th Annual Symposium on Foundations of Computer Science (FOCS)},
  pages={273--286},
  year={2019},
  organization={IEEE}
}

@inproceedings{Sberlo18,
  title={On the performance of {Reed-Muller} codes with respect to random errors and erasures},
  author={Sberlo, O. and Shpilka, A.},
  booktitle={Proceedings of the Fourteenth Annual ACM-SIAM Symposium on Discrete Algorithms},
  pages={1357--1376},
  year={2020},
  organization={SIAM}
}

@ARTICLE{YA18,
  author={Ye, Min and Abbe, Emmanuel},
  journal={IEEE Transactions on Information Theory}, 
  title={Recursive Projection-Aggregation Decoding of {Reed-Muller} Codes}, 
  year={2020},
  volume={66},
  number={8},
  pages={4948-4965}}

@ARTICLE{reeves,
  author={Reeves, Galen and Pfister, Henry D.},
  journal={IEEE Transactions on Information Theory}, 
  title={Reed–{M}uller Codes on {BMS} Channels Achieve Vanishing Bit-Error Probability for All Rates Below Capacity}, 
  year={2023},
  volume={},
  number={},
  pages={1-1},
  doi={10.1109/TIT.2023.3286452}}

@article{friedgut1996every,
  title={Every monotone graph property has a sharp threshold},
  author={Friedgut, E. and Kalai, G.},
  journal={Proceedings of the American mathematical Society},
  volume={124},
  number={10},
  pages={2993--3002},
  year={1996}
}

@article{Bourgain97,
  title={Influences of variables and threshold intervals under group symmetries},
  author={Bourgain, J. and Kalai, G.},
  journal={Geometric and Functional Analysis},
  volume={7},
  number={3},
  pages={438--461},
  year={1997},
  publisher={Springer}
}

@article{Yekhanin12,
  title={Locally decodable codes},
  author={Yekhanin, S.},
  journal={Foundations and Trends{\textregistered} in Theoretical Computer Science},
  volume={6},
  number={3},
  pages={139--255},
  year={2012},
  publisher={Now Publishers, Inc.}
}

@article{Shannon48,
  title={A mathematical theory of communication},
  author={Shannon, C. E.},
  journal={Bell system technical journal},
  volume={27},
  number={3},
  pages={379--423},
  year={1948},
  publisher={Wiley Online Library}
}

@article{Hamming50,
  title={Error detecting and error correcting codes},
  author={Hamming, R. W.},
  journal={The Bell system technical journal},
  volume={29},
  number={2},
  pages={147--160},
  year={1950},
  publisher={Nokia Bell Labs}
}

@article{Kaufman12,
  title={Weight distribution and list-decoding size of {Reed--Muller} codes},
  author={Kaufman, T. and Lovett, S. and Porat, E.},
  journal={IEEE Transactions on Information Theory},
  volume={58},
  number={5},
  pages={2689--2696},
  year={2012},
  publisher={IEEE}
}

@ARTICLE{Samorod18,
  author={Samorodnitsky, Alex},
  journal={IEEE Transactions on Information Theory}, 
  title={An Upper Bound on $\ell_q$ Norms of Noisy Functions}, 
  year={2020},
  volume={66},
  number={2},
  pages={742-748}}

@article{Costello07,
  title={Channel coding: {T}he road to channel capacity},
  author={Costello, D. J. and Forney, G. D.},
  journal={Proceedings of the IEEE},
  volume={95},
  number={6},
  pages={1150--1177},
  year={2007},
  publisher={IEEE}
}

@article{DBLP:journals/tit/AlonKKLR05,
	author    = {Noga Alon and
	Tali Kaufman and
	Michael Krivelevich and
	Simon Litsyn and
	Dana Ron},
	title     = {Testing {Reed-Muller} codes},
	journal   = {{IEEE} Trans. Inf. Theory},
	volume    = {51},
	number    = {11},
	pages     = {4032--4039},
	year      = {2005}
}

@inproceedings{DBLP:conf/focs/BhattacharyyaKSSZ10,
	author    = {Arnab Bhattacharyya and
	Swastik Kopparty and
	Grant Schoenebeck and
	Madhu Sudan and
	David Zuckerman},
	title     = {Optimal Testing of {Reed-Muller} Codes},
	booktitle = {51th Annual {IEEE} Symposium on Foundations of Computer Science, {FOCS}
	2010, October 23-26, 2010, Las Vegas, Nevada, {USA}},
	pages     = {488--497},
	publisher = {{IEEE} Computer Society},
	year      = {2010}
}

@article{DBLP:journals/siamcomp/HaramatySS13,
	author    = {Elad Haramaty and
	Amir Shpilka and
	Madhu Sudan},
	title     = {Optimal Testing of Multivariate Polynomials over Small Prime Fields},
	journal   = {{SIAM} J. Comput.},
	volume    = {42},
	number    = {2},
	pages     = {536--562},
	year      = {2013}
}

@article{DBLP:journals/siamcomp/KaufmanR06,
	author    = {Tali Kaufman and
	Dana Ron},
	title     = {Testing Polynomials over General Fields},
	journal   = {{SIAM} J. Comput.},
	volume    = {36},
	number    = {3},
	pages     = {779--802},
	year      = {2006}
}

@article{DBLP:journals/rsa/JutlaPRZ09,
	author    = {Charanjit S. Jutla and
	Anindya C. Patthak and
	Atri Rudra and
	David Zuckerman},
	title     = {Testing low-degree polynomials over prime fields},
	journal   = {Random Struct. Algorithms},
	volume    = {35},
	number    = {2},
	pages     = {163--193},
	year      = {2009}
}

@article{Shamir79,
  title={How to share a secret},
  author={Shamir, A.},
  journal={Communications of the ACM},
  volume={22},
  number={11},
  pages={612--613},
  year={1979},
  publisher={ACM}
}

@incollection{BF90,
  title={Hiding instances in multioracle queries},
  author={Beaver, D. and Feigenbaum, J.},
  booktitle={STACS 90},
  pages={37--48},
  year={1990},
  publisher={Springer}
}

@inproceedings{Gasarch04,
  title={A survey on private information retrieval},
  author={Gasarch, W.},
  booktitle={Bulletin of the EATCS},
  year={2004},
  organization={Citeseer}
}

@article{ALMSS98,
  title={Proof verification and the hardness of approximation problems},
  author={Arora, S. and Lund, C. and Motwani, R. and Sudan, M. and Szegedy, M.},
  journal={Journal of the ACM (JACM)},
  volume={45},
  number={3},
  pages={501--555},
  year={1998},
  publisher={ACM}
}

@article{bogdanov-viola,
  author    = {A. Bogdanov and
               E. Viola},
  title     = {Pseudorandom Bits for Polynomials},
  journal   = {SIAM J. Comput.},
  volume    = {39},
  number    = {6},
  year      = {2010},
  pages     = {2464-2486}
}

@inproceedings{BFL90,
  title={Nondeterministic exponential time has two-prover interactive protocols},
  author={Babai, L. and Fortnow, L. and Lund, C.},
  booktitle={Foundations of Computer Science, 1990. Proceedings., 31st Annual Symposium on},
  pages={16--25},
  year={1990},
  organization={IEEE}
}

@article{Sha92,
  title={{IP= PSPACE}},
  author={Shamir , A.},
  journal={Journal of the ACM (JACM)},
  volume={39},
  number={4},
  pages={869--877},
  year={1992},
  publisher={ACM}
}

@inproceedings{barak2012making,
  author    = {B. Barak and
               P. Gopalan and
               J. Hastad and
               R. Meka and
               P. Raghavendra and
               D. Steurer},
  title     = {Making the Long Code Shorter},
  booktitle = {53rd Annual {IEEE} Symposium on Foundations of Computer Science, {FOCS}
               2012, New Brunswick, NJ, USA, October 20-23, 2012},
  year      = {2012},
  pages     = {370--379}
}

@Article{Razborov,
  author =   {A. A. Razborov},
  title =    {Lower bounds on the size of bounded depth circuits over a
complete basis with logical addition},
  journal =      {Math. Notes},
  year =     {1987},
  OPTkey =   {},
  volume =   {41},
  number =   {4},
  pages =    {333-338}
}

@inproceedings{DBLP:conf/focs/BeimelIKR02,
  author    = {Amos Beimel and
               Yuval Ishai and
               Eyal Kushilevitz and
               Jean{-}Fran{\c{c}}ois Raymond},
  title     = {Breaking the $O(n^{1/(2k-1)})$ Barrier for Information-Theoretic Private
               Information Retrieval},
  booktitle = {43rd Symposium on Foundations of Computer Science {(FOCS} 2002), 16-19
               November 2002, Vancouver, BC, Canada, Proceedings},
  pages     = {261--270},
  publisher = {{IEEE} Computer Society},
  year      = {2002}
}

@article{DBLP:journals/jcss/Ta-ShmaZS06,
  author    = {Amnon Ta{-}Shma and
               David Zuckerman and
               Shmuel Safra},
  title     = {Extractors from {Reed-Muller} codes},
  journal   = {J. Comput. Syst. Sci.},
  volume    = {72},
  number    = {5},
  pages     = {786--812},
  year      = {2006}
}

@inproceedings{HSS,
  author    = {Jan Hazla and
               Alex Samorodnitsky and
               Ori Sberlo},
  editor    = {Samir Khuller and
               Virginia Vassilevska Williams},
  title     = {On codes decoding a constant fraction of errors on the {BSC}},
  booktitle = {{STOC} '21: 53rd Annual {ACM} {SIGACT} Symposium on Theory of Computing,
               Virtual Event, Italy, June 21-25, 2021},
  pages     = {1479--1488},
  publisher = {{ACM}},
  year      = {2021}
}

@article{Gallager65,
  title={A simple derivation of the coding theorem and some applications},
  author={Gallager, R},
  journal={IEEE Transactions on Information Theory},
  volume={11},
  number={1},
  pages={3--18},
  year={1965},
  publisher={IEEE}
}

@misc{gowers2023conjecturemarton,
      title={On a conjecture of {Marton}}, 
      author={W. T. Gowers and Ben Green and Freddie Manners and Terence Tao},
      year={2023},
      eprint={2311.05762},
      archivePrefix={arXiv},
      primaryClass={math.NT},
      url={https://arxiv.org/abs/2311.05762}, 
}

@misc{rao2022list,
      title={On List Decoding Transitive Codes From Random Errors}, 
      author={Anup Rao and Oscar Sprumont},
      year={2022},
      eprint={2202.00240},
      archivePrefix={arXiv},
      primaryClass={cs.IT}
}

@ARTICLE{Geiselhart21,
  author={Geiselhart, Marvin and Elkelesh, Ahmed and Ebada, Moustafa and Cammerer, Sebastian and ten~Brink, Stephan},
  journal={IEEE Transactions on Communications}, 
  title={Automorphism Ensemble Decoding of {Reed–Muller} Codes}, 
  year={2021},
  volume={69},
  number={10},
  pages={6424-6438}}

@inproceedings{Lian20,
  title={Decoding {Reed--Muller} codes using redundant code constraints},
  author={Lian, Mengke and H{\"a}ger, Christian and Pfister, H. D.},
  booktitle={2020 IEEE International Symposium on Information Theory (ISIT)},
  pages={42--47},
  year={2020},
  organization={IEEE}
}

@inproceedings{Fathollahi21,
  title={Sparse multi-decoder recursive projection aggregation for {Reed-Muller} codes},
  author={Fathollahi, Dorsa and Farsad, Nariman and Hashemi, Seyyed Ali and Mondelli, Marco},
  booktitle={2021 IEEE International Symposium on Information Theory (ISIT)},
  pages={1082--1087},
  year={2021},
  organization={IEEE}
}

@inproceedings{Calderbank2010reed,
  title={{Reed-Muller sensing matrices and the LASSO}},
  author={Calderbank, Robert and Jafarpour, Sina},
  booktitle={International Conference on Sequences and Their Applications},
  pages={442--463},
  year={2010},
  organization={Springer}
}

@ARTICLE{Calderbank10,
  author={Calderbank, Robert and Howard, Stephen and Jafarpour, Sina},
  journal={IEEE Journal of Selected Topics in Signal Processing}, 
  title={Construction of a Large Class of Deterministic Sensing Matrices That Satisfy a Statistical Isometry Property}, 
  year={2010},
  volume={4},
  number={2},
  pages={358-374}}

@ARTICLE{Barg15,
  author={Barg, Alexander and Mazumdar, Arya and Wang, Rongrong},
  journal={IEEE Transactions on Information Theory}, 
  title={Restricted Isometry Property of Random Subdictionaries}, 
  year={2015},
  volume={61},
  number={8},
  pages={4440-4450}}

@ARTICLE{ASY21,
  author={Abbe, Emmanuel and Shpilka, Amir and Ye, Min},
  journal={IEEE Transactions on Information Theory}, 
  title={{Reed–Muller Codes: Theory and Algorithms}}, 
  year={2021},
  volume={67},
  number={6},
  pages={3251-3277}}

@inproceedings{BhandariHS022,
  author    = {Siddharth Bhandari and
               Prahladh Harsha and
               Ramprasad Saptharishi and
               Srikanth Srinivasan},
  editor    = {Shachar Lovett},
  title     = {Vanishing Spaces of Random Sets and Applications to {Reed-Muller} Codes},
  booktitle = {37th Computational Complexity Conference, {CCC} 2022, July 20-23,
               2022, Philadelphia, PA, {USA}},
  series    = {LIPIcs},
  volume    = {234},
  pages     = {31:1--31:14},
  publisher = {Schloss Dagstuhl - Leibniz-Zentrum f{\"{u}}r Informatik},
  year      = {2022}
}


\bibliographystyle{amsplain}

\end{document}